\documentclass[aps,pra,twocolumn,superscriptaddress,nofootinbib,10pt]{revtex4-2}

\usepackage{times}

\usepackage{amsmath,amsfonts,amssymb,amsthm, bbm,dsfont,mathtools}
\usepackage{enumitem}
\usepackage{xcolor}
\usepackage{graphicx}   
\usepackage{natbib}
\usepackage{comment}
\usepackage{algpseudocode}
\usepackage[colorlinks=true,linkcolor=magenta,citecolor=magenta]{hyperref}
\usepackage[ruled]{algorithm2e}

\usepackage{diagbox}
\usepackage{makecell}

\usepackage{appendix}

\usepackage{titlesec}
\newcommand\largesection{%
	\titleformat{\section}
	{\normalfont\Large\bfseries\filcenter}{\thesection}{1em}{}
}
\newcommand\largesubsection{%
	\titleformat{\subsection}
	{\normalfont\large\bfseries\filcenter}{\thesubsection}{1em}{}
}

\newcommand\largesubsubsection{%
	\titleformat{\subsubsection}
	{\normalfont\normalfont\bfseries\filcenter}{\thesubsubsection}{1em}{}
}

\usepackage[nameinlink,capitalise]{cleveref}
\crefname{equation}{Eq.}{Eqs.}

\usepackage{multirow}

\usepackage{placeins}
\usepackage[margin=2.0cm]{geometry}
\allowdisplaybreaks

\newtheorem{theorem}{Theorem}

\newtheorem{lemma}[theorem]{Lemma}

\newcommand{\myleft}{\mathopen{}\mathclose\bgroup\left}
\newcommand{\myright}{\aftergroup\egroup\right}

\DeclareMathOperator{\Tr}{Tr}

\DeclareMathOperator{\supp}{supp}

\DeclareMathOperator{\Ad}{Ad}

\DeclareMathOperator\tr{Tr}

\newcommand{\RB}{\mathrm{RB}}
\newcommand{\FT}{\mathrm{FT}}
\newcommand{\hatM}{\widehat{\mathcal{M}}}

\newcommand{\CC}{\mathbb{C}}

\newcommand{\KK}{\mathbb{K}}

\newcommand{\EE}{\mathbb{E}}

\newcommand{\mc}[1]{\mathcal{#1}}

\newcommand{\GG}{\mc{G}}
\newcommand{\1}{\mathds{1}}

\newcommand{\norm}[1]{\left\Vert #1 \right\Vert} 

\newcommand{\ket}[1]{\left.\left|{#1}\right.\right\rangle}

\newcommand{\bra}[1]{\left.\left\langle{#1}\right.\right|}

\newcommand{\braket}[2]{\left\langle #1 \middle| #2 \right\rangle}

\newcommand{\ketbra}[2]{\ket{#1} \!\! \bra{#2}}

\newcommand{\sandwich}[3]
  {\left\langle  #1 \right| #2 \left| #3 \right\rangle}

\newcommand{\kett}[1]{|{#1}{\rangle\!\rangle}}
\newcommand{\braa}[1]{{\langle\!\langle}{#1}|}
\newcommand{\kettbraa}[2]{|{#1}{\rangle\!\rangle}\!{\langle\!\langle}{#2}|}

\newcommand{\loc}{\mathrm{loc}}

\newcommand{\tn}[1]{^{\otimes#1}} 

\newcounter{example}[section]

\newcommand{\wh}[1]{\widehat{#1}}

\newcommand{\abs}[1]{\lvert #1 \rvert}

\newenvironment{aligneq}{\begin{equation}\begin{aligned}}{\end{aligned}\end{equation}}

\newcommand{\natM}{\mc {M^\natural}}
\newcommand{\triv}{\mathrm{triv}}
\newcommand{\adj}{\mathrm{adj}}
\newcommand{\avg}{\mathrm{avg}}
\newcommand{\ort}{\mathrm{ort}}
\newcommand{\inv}{\mathrm{inv}}
\newcommand{\eend}{\mathrm{end}}
\renewcommand{\CC}{\mathcal{C}}

\renewcommand{\ket}[1]{\left.\left|{#1}\right.\right\rangle}
\renewcommand{\bra}[1]{\left.\left\langle{#1}\right.\right|}
\renewcommand{\ketbra}[2]{\ket{#1} \!\! \bra{#2}}
\newcommand{\braakett}[2]{{\langle\!\langle}{#1}|#2{\rangle\!\rangle}}

\usepackage{tabularx} 
\newcounter{protocol}
\newenvironment{protocol}[1]
  {\par\addvspace{\topsep}
   \noindent
   \tabularx{\linewidth}{@{} X @{}}
    \hline
    \refstepcounter{protocol}\textbf{Protocol:} #1 \\
    \hline}
  { \\
    \hline
   \endtabularx
   \par\addvspace{\topsep}}

\definecolor{citegreen}{RGB}{0,165,0}

\definecolor{jens}{rgb}{0.1,0.5,0.1}

\definecolor{ingo}{rgb}{1,.2,.4}

\definecolor{mulberry}{rgb}{0.77, 0.29, 0.55}

\definecolor{myred}{RGB}{189,0,0}
\definecolor{mygreen}{RGB}{46,139,87}

\makeatletter 
    
\renewcommand\onecolumngrid{
\do@columngrid{one}{\@ne}%
\def\set@footnotewidth{\onecolumngrid}
\def\footnoterule{\kern-6pt\hrule width 1.5in\kern6pt}%
}

\renewcommand\twocolumngrid{
        \def\footnoterule{
        \dimen@\skip\footins\divide\dimen@\thr@@
        \kern-\dimen@\hrule width.5in\kern\dimen@}
        \do@columngrid{mlt}{\tw@}
}%

\makeatother   

\begin{document}

\title{Noise-mitigated randomized measurements and self-calibrating shadow estimation}

\author{E. Onorati}
\email[email: ]{emilio.onorati@tum.de}
\affiliation{Zentrum Mathematik, Technische Universit{\"a}t M{\"u}nchen, 85748 Garching, Germany}
\affiliation{Dahlem Center for Complex Quantum Systems, Freie Universit{\"a}t Berlin, 14195 Berlin, Germany}

\author{J. Kitzinger}
\affiliation{Dahlem Center for Complex Quantum Systems, Freie Universit{\"a}t Berlin, 14195 Berlin, Germany}

\author{J. Helsen}
\affiliation{QuSoft, Centrum Wiskunde and Informatica (CWI), Amsterdam, The Netherlands}

\author{M. Ioannou}
\affiliation{Dahlem Center for Complex Quantum Systems, Freie Universit{\"a}t Berlin, 14195 Berlin, Germany}

\author{A. H. Werner}
\affiliation{Department of Mathematical Sciences, University of Copenhagen, 2100 Copenhagen, Denmark}

\author{I. Roth}
\affiliation{Quantum Research Centre, Technology Innovation Institute (TII), Abu Dhabi}

\author{J. Eisert}
\affiliation{Dahlem Center for Complex Quantum Systems, Freie Universit{\"a}t Berlin, 14195 Berlin, Germany}
\affiliation{Helmholtz-Zentrum Berlin f{\"u}r Materialien und Energie, 14109 Berlin, Germany}
\affiliation{Fraunhofer Heinrich Hertz Institute, 10587 Berlin, Germany}

\begin{abstract}
Randomized measurements are increasingly appreciated as powerful tools to estimate properties of quantum systems, e.g., in the characterization of hybrid classical-quantum computation. On many platforms they constitute natively accessible measurements, serving as the building block of prominent schemes like shadow estimation. In the real world, however, the implementation of the random gates at the core of these schemes is susceptible to various sources of noise and imperfections, strongly limiting the applicability of protocols. To attenuate the impact of this shortcoming, in this work we introduce an error-mitigated method of randomized measurements, giving rise to a robust shadow estimation procedure. On the practical side, we show that error mitigation and shadow estimation can be carried out using the same session of quantum experiments, hence ensuring that we can address and mitigate the noise affecting the randomization measurements. Mathematically, we develop a picture derived from Fourier-transforms to connect randomized benchmarking and shadow estimation. We prove rigorous performance guarantees and show the functioning using comprehensive numerics. More conceptually, we demonstrate that, if properly used, easily accessible data from randomized benchmarking schemes already provide such valuable diagnostic information to inform about the noise dynamics and to assist in quantum learning procedures.
\end{abstract}

\maketitle

Randomized measurements are a ubiquitous tool in quantum technologies. They combine the application of random quantum gates with a measurement, often in the computational basis that is native to a given technological platform. From the practitioner's perspective, thus, they are the naturally accessible resource on digital and analog quantum computing platforms. 
From a theoretical perspective, they effectively implement measurement frames with desirable encoding properties of the quantum information of the device's pre-measurement state. For this reason, randomized measurements have become the foundation of many important quantum characterization protocols, for benchmarking, certification, estimation and tomography~\cite{BenchmarkingReview,PRXQuantum.2.010201,Randomizedtoolbox, scott2006tight, KnillBenchmarking, EmersonRB, MagGamEmer2,Efficient, Shadows,RandomSequences,RandomRenyi,RobustShadows,frameworkRB}. 
A possibly most prominent incarnation of such schemes are  estimation protocols based on so-called~\emph{classical shadows} of quantum states~\cite{Shadows,Randomizedtoolbox,RandomRenyi, PainiKalev:2019:Shadows,RobustShadows}.
These schemes have been shown \cite{Shadows} to often have close-to-optimally efficiency in the simultaneous estimation of multiple observables in terms of the sampling complexity -- a central desideratum in order to render the schemes practically applicable in today's platforms.
At the same time, they feature in instances of hybrid quantum-classical approaches~\cite{McClean_2016,Variational,Kandala, Wiersema2022Interfaces}, where a quantum algorithm is augmented by a surrounding classical algorithm, and randomized measurements serve as a \emph{quantum-classical interfaces}.

\begin{figure}
\includegraphics[width=1\columnwidth]{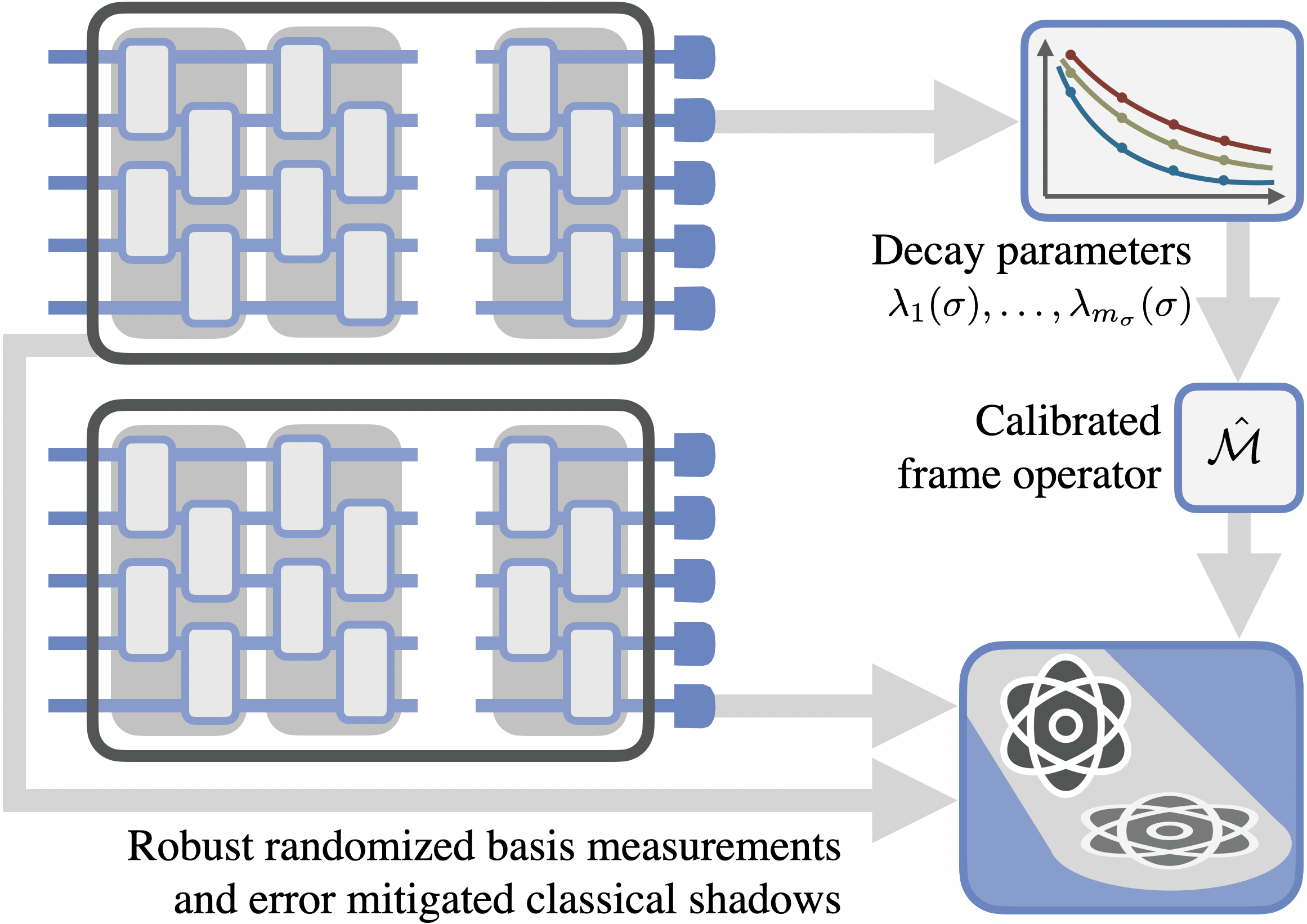}
\caption{An unknown state is subjected to random Clifford sequences, followed by measurements in the computational basis. Using randomized benchmarking, decay parameters are being estimated to establish the calibrated frame operator~$\hatM$.
The calibration procedure can be part of the estimation or run separately, in four variants.}
\label{figure}
\end{figure}

A key insight in quantum systems identification is that for schemes to be experimentally feasible, \emph{robustness} against experimental imperfections is crucial.
This not only 
means that they must be able to tolerate \emph{state-preparation and measurement} (SPAM) errors, but also to address implementation inaccuracy 
of the quantum gates assisting the measurement scheme. 
In the context of shadow estimation, the effect of noise affecting the randomized measurements can -- under certain assumptions -- be canceled out in the classical post-processing~\cite{RobustShadows,koh2022ClassicalShadowsNoise, Brieger21CompressiveGateSet}.  
These approaches can be seen as generalization of read-out error-mitigation using bit-flip randomization \cite{Karalekas:2020:Quantumclassical,vanDerBerg:2022:Model-free}.  
Importantly, to inform the classical post-processing in the former approaches, the effect of the noise on the estimation needs to be characterized in a separate \emph{calibration experiment} that itself can be affected by SPAM errors. 
The recent Ref.~\cite{zhao2023grouptheoretic} overcomes the need of calibration when there are observables that encode symmetries of the system and thus have a-priori known values.
A different line of work focuses on the distinct goal of mitigating the noise in the state-preparation using non-linear shadow estimation \cite{ShadowDistillation, PhysRevX.11.041036, hu2022logical} or existing error-mitigation techniques \cite{Jnane:2024:MitigatedShadows}. 

In this work, we aim at making randomized measurements substantially more robust: 
We address the problem of mitigating the noise of the quantum gates randomizing the measurements by directly making use of the arguably most common and successful technique for gate-noise characterization: \emph{randomized benchmarking} (RB) \cite{KnillBenchmarking, EmersonRB,frameworkRB}. 
We start from the observation that  
non-commutative generalized Fourier transforms \cite{GowersHatami} 
naturally describe both, noisy implementation of randomized bases measurement and the general gate-dependent noise theory of RB \cite{Merkel2021randomized, frameworkRB}. 
This perspective indicates how to use the approximate characterization of Fourier transform from RB experiments in the context of classical shadows. 

Following this perspective, we devise multiple new protocols for robust shadow estimation where noise-effects are mitigated using information obtained from RB routines. Building our approach upon RB routines with sequences of random gates of different length, has multiple attractive features: 
i) Making use of common RB, the data required for the noise-mitigation might already be acquired during tune-up and calibration of the gates.  
ii) Going beyond this, we develop self-calibrating protocols where one makes use of the same experimental data for both the calibration and the estimation. 
The self-calibrating protocols can in practice dramatically reduce the required experimental resources. 
Our protocols inherit iii) the SPAM-robustness of RB when calibrating against the noise affecting the gates,  and iv) the efficiency and scalability of RB.

More specifically, 
we present four variants of our approach: one based on the multi-qubit Clifford group,  one based on parallel application of single-qubit Clifford unitaries, and two protocols based on the CNOT-dihedral subgroup. 
We show that under common assumptions, the noise parameters entering randomized measurements can be precisely related to exponential decay parameters of RB protocols.  
Thus, we theoretically guarantee that using RB decay parameters, noise and errors affecting the estimation protocol can be perfectly mitigated in expectation.  
We further find a favourable scaling of the efficiency of noise-mitigated estimation schemes using multi-qubit rotations that is comparable to their noise-less counterparts, if the noise strength itself does not scale with the system sizes.
We derive rigorous guarantees under the assumption of a simple gate-independent noise model. 
However, the perspective from the general theory of RB already 
makes the required approximate assumptions to guarantee functioning in the presence of gate-dependent noise apparent. 
In numerical simulation, we validate the resulting expectation that also in the presence of `typical' gate-dependent noise, our protocols significantly mitigate noise-induced estimation biases. 

Using longer random sequences for ``washing out'' gate-dependent noise effects can further improve the stability of classical shadows, excluding pathological settings identified by Ref.~~\cite{brieger2023StabilityClassicalShadows}. 
From the perspective of error-mitigation techniques \cite{MitigationReview}, our scheme can also be seen as following the paradigm of zero-noise extrapolation \cite{PhysRevLett.119.180509,li2017efficient}, here in form of the noise amplification at the heart of RB techniques. 


\paragraph*{Mindset of the protocol.} 
Randomized basis measurements provide a faithful interface between a quantum state and classical data. 
The data can be post-processed or used for probabilistically re-preparing the state on a quantum device. To this end, one implements measurements in a random basis by rotating the state with a unitary, say, randomly drawn from a group $\mc G$, before read-out in the computational basis of the device. 
In the absence of noise, the effective measurement channel, or \emph{frame operator} of the implemented \emph{informationally complete positive operator-valued measure} (IC-POVM), reads 
\begin{equation}\label{def:frame_operator}
  \mc M_{\mathrm{ideal}} \coloneqq \EE_\GG \Ad^\dagger(g) B \Ad(g)\,,
\end{equation}
where $\Ad(g) = U(g) \otimes U(g)^\ast$ denotes the action of $g\in \mc G$ as a unitary channel, $B=\sum_{b \in \{0,1\}^{\times n}}\ketbra{b,b}{b,b}$ is the dephasing channel in the computational basis, and the average is taken over the Haar measure. 
This effective measurement channel needs to be inverted in order to reconstruct the information about the quantum state in the classical post-processing. 
On an actual quantum device, the implementation of the rotation before the basis measurement will, however, inevitably suffer from noise. This is particularly daunting if implementing unitaries in $\mc G$ requires long circuits with entangling gates. 
A fairly general noise model can be written in terms of an implementation map $\phi$ that associates to an intended operation $g\in \mc G$ the corresponding quantum channel $\phi(g)$ that actually happens on the device. 
If instead of the ideal frame operator, we now invert its analog containing the noise 
\begin{equation}\label{def:noisy_frame_operator}
  \mc M \coloneqq \EE_\GG \Ad^\dagger(g) B \phi(g)\, ,
\end{equation}
we are effectively canceling out the noise and arrive again at an unbiased estimator.
Now from a mathematical perspective, the expression \cref{def:noisy_frame_operator} for the `noisy' frame operator can be understood as the non-commutative Fourier transform~(FT) of $\phi$ acting on $B$, that is,
\begin{equation}\label{eq:M_FT_connection}
    \mc M = \phi_\FT[\Ad](B).
\end{equation}
This can be thought of as a formal way of referring to `noisy channel twirls'. 
The mindset behind our work is that exactly this Fourier transform for important gate sets is routinely characterized by nearly every quantum computing experiment on the planet, namely, when running RB experiments~\cite{Merkel2021randomized,frameworkRB}.
More precisely, RB characterizes the dominant eigenvalues of these Fourier transforms as the decay parameters of signals that fade for longer and longer random sequences of gates drawn from~$\mc G$. 
Under the same assumptions as those required for output data of RB experiments to be related to average gate fidelities, these decay parameters already provide a full description of the entire operator $\phi_\FT[\Ad]$.  
Thus, RB experiments should in many settings naturally provide the required information for mitigating noise in randomized basis measurements.

\paragraph*{Randomized benchmarking routine as starting point.}
Randomized benchmarking protocols are a family of techniques built on sequences of random gates to extract data (usually, average gate fidelities of the implemented gate set). That is, given a set of gates $\{g\}_{g\in \mc G}$ that have physical implementations $\{\phi(g)\}_{g\in \mc G}$ that we want to characterize, an RB protocol samples gate sequences  $\vec{g}(m)=(g_1,g_2,\dots,g_m)$.
from the gate set at random, computes a global inverse $g_{\mathrm{inv}} = (g_m\cdots g_1)^{-1}$, and then estimates the probability
\begin{equation}
    \mc Q (\rho_0,E, \vec{g}) \coloneqq \Tr [E \, \phi(g_{\mathrm{inv}})\phi(g_m) \cdots \phi(g_2) \phi(g_1) \, (\rho_0)], 
\end{equation}
where $E$ and $\rho$ are a measurement POVM element and input state, respectively. Note that $g_{\mathrm{inv}}$ does not need to be physically implemented, it can also be accounted for in the classical post-processing through so-called filter functions~\cite{helsen2019new,PhysRevLett.123.060501,frameworkRB, MarkusHeinrich}, which we will use in the classical post-processing stage of our schemes. 
For different circuits lengths~$m$, we can repeat this measurement for several random sequences and average to obtain a sample average~$\mc Q_\avg$, which is fit to an exponential model with decay parameters $\{\lambda_\alpha\}_\alpha$ that constitute the figures of merit of an RB protocol,
\begin{equation}
    \mc Q_\avg(m) = \sum_{\alpha} y_\alpha(E, \rho) \lambda_\alpha^m.
\end{equation}
Conveniently, by making use of the filtering functions, we can also extract each $\lambda_\alpha$ separately.

One of the major features of RB protocols is that they are often sample-efficient in achieving multiplicative precision, and robust against SPAM errors. RB is therefore an ideal candidate to provide us with some   characterization of noisy randomized measurements.

\paragraph*{RB estimates for noisy frame operators.}
In reasonably good experimental implementations, the implementation map $\phi$ is close to a representation. 
As a consequence, its Fourier transform maintains the distinctive structure of the Fourier transform of a representation:  
Breaking down~$\phi_\FT[\Ad]$ for each irreducible representation $\alpha$ of $\GG$ occurring in~$\Ad$,
we write using standard isomorphisms  $\phi_\FT[\alpha] \coloneqq \EE_\GG \, \alpha^\ast (g) \otimes \phi (g)$.
Let $m_\alpha$ be the multiplicity of $\alpha$ in $\Ad$. 
The operator $\phi_\FT[\alpha]$  (being a perturbation of the  projection $\Ad_\FT[\alpha]$) has $m_\alpha$ eigenvalues close to $1$ and the remaining eigenvalues close to zero \cite{frameworkRB}.
In a nutshell, RB signals are described by matrix powers of $\phi_\FT[\alpha]$ \cite{Merkel2021randomized,frameworkRB}, thus, the observed decay parameters correspond to the dominant eigenvalue(s) 
$\lambda_\alpha$ of $\phi_\FT[\alpha]$. 
In case of $m_\alpha>1$, we pick a representative value for the dominant eigenvalues, e.g. taking the average thereof. 
Using \cref{eq:M_FT_connection}, we can use $\lambda_\alpha$ as follows to arrive at an estimate for $\mc M$ in the presence of noise: 
For each $\alpha$ in the irreducible decomposition appearing in $\Ad$, we define
\begin{equation}\label{def:phi_estimator}
	\wh \phi_{\FT} [\alpha] 
	\coloneqq 
	\lambda_{\alpha} \Ad_{\FT}[\alpha]\, .
\end{equation}
For this estimate to be accurate, we implicitly assume that the noise mainly effects the spectrum of the channel twirl but preserves its subspace structure. 

We will from now on restrict to multiplicity-free decompositions of $\Ad$ for simplicity. 
Evaluating \cref{eq:M_FT_connection} from the estimator \cref{def:phi_estimator} may require the substitution of the coefficients $\lambda_\alpha$, that is, associating $\Ad[\alpha]$ with a decay parameter $\lambda_{\alpha'}$ of a group~$\GG'$ which may possibly be different from~$\GG$. 
The reason for this is the presence of $B$ in $\hatM$, which produces a shift in the invariant subspace of $\phi_\FT[\Ad]$.
Namely, we construct
\begin{equation}\label{eq:hybrid_framework_estimation}
	\hatM = \sum_{(\alpha,\alpha')} \lambda_{\alpha'} \frac{\Tr[\Pi_\alpha B]}{\Tr[\Pi_\alpha]}\Pi_\alpha,
\end{equation}
where $\Pi_\alpha$ is the projector onto the irrep~$\alpha$ of $\GG$ and $\lambda_{\alpha'}$ is the RB decay parameter of an irrep~$\alpha'$ of $\GG'$.

\paragraph*{Frame operator for robust shadow estimation.} 
As an application of the general approach explained in the previous section, we can construct noise-mitigated shadow estimation protocols.  
We start with global Clifford measurements, i.e., we set $\GG$ to be the multi-qubit Clifford group, whose $\Ad$ representation splits into two inequivalent irreducible components. The noise-free frame operator in this setting is given by~\cite{Shadows}
\begin{equation} \label{eq:ideal_frame}
	\mc M_{\mathrm{ideal}}  = \Pi_\triv + \Pi_\adj/(d+1),
\end{equation}
where $d$ is the dimension of the system, $\Pi_\triv$ is the projector onto the fully mixed state, and $\Pi_\adj=\1-\Pi_\triv$ is its complement.
The frame operator is altered by noise 
and thus must be corrected to yield reliable shadow estimates. 
Such a correction has been proposed in Refs.~\cite{koh2022ClassicalShadowsNoise,RobustShadows} (which does so through a separate calibration experiment).

We propose a strategy (see \hyperlink{protocol}{Protocol}) to estimate the frame operator by using an RB protocol drawing gates from the CNOT-dihedral group, a finite group generated by $\mathbb{K}= \langle \mathrm{CNOT}, S, X \rangle$, with $S$ being the phase gate, to extract the decay parameter $\lambda_Z$. 
For gate-independent noise, $\lambda_Z$ coincides with effective depolarizing strength of a noise channel restricted to the subspace of the computational basis (spanned by Pauli-$Z$ operators) and might be refered to as the effective dephasing strength -- cfr.\ \cref{eq:Z-par_independent_noise}.
Single-qubit dihedral and CNOT-dihedral RB has been proposed in Refs.~\cite{CNOTDihedral,carignan2015characterizing}, respectively.  
We here present a modified version for CNOT-dihedral RB using filter functions. 
In the spirit of \cref{eq:hybrid_framework_estimation}, we define an estimate for the frame operator in terms of the decay parameter~$\lambda_Z$ as 
\begin{equation}\label{eq:inverted_Mnat}
    \hatM = \Pi_\triv + \frac{\lambda_Z}{d+1} \Pi_\adj\,.
\end{equation}
We give a more careful justification for using $\lambda_Z$ arising in CNOT-dihedral RB in \cref{app:RB_dihedral}. 
In particular, we show that in the case of gate-independent noise we arrive at perfect noise-cancellation in expectation. 

The protocol explained above still requires a separate RB experiment for the calibration. 
But we can actually devise an (informationally-complete) estimation protocol that is self-calibrating with a small adaptation: We pick a single Clifford gate $c$ at random and implement it as the first gate of the circuit.
The remaining $m$ gates will again be drawn from the CNOT-dihedral group. 
Since this circuit corresponds to a random Clifford operator thanks to the first gate, we can use the measurement data from this RB experiment to directly construct the shadows. 
At the same time, we can use the same set of measurement outcomes $\{b'\}$ to classically compute $(d+1)\Tr[E \, \Pi_\adj \Ad^\dagger(g_\eend )(b')]$ (where $g_\eend$ is the group element resulting from the composition of the whole circuit) and use these quantities to classically extract $\lambda_Z$ as the decay parameter of the fitting model for the sample average, that is,
\begin{equation}\label{eq:fitting_model_for_lambda}
    \mc Q(m+1)
    = 
    \Tr[E \, \Pi_\adj (\rho)] \, \lambda_Z^{m+1}.
\end{equation}
We will then use the extracted parameter to calibrate the frame operator,
\begin{equation}\label{eq:RB_calibrated_frame_circuit}
    \hatM_{m+1} \coloneqq \Pi_\triv +  \frac{\lambda_Z(\Lambda)^{m+1}}{d+1}  \Pi_\adj ,
\end{equation} 
and produce the noise-mitigated classical shadow from every outcome $b'$ resulting from a random circuit of length $m+1$,
\begin{equation}
    \hat \rho_{m+1} = \hatM^{-1}_{m+1}\Ad^\dagger(g_\eend) \kett{b'} .
\end{equation}
In this way, we can construct shadows \emph{and} calibrate them using the same data from a single quantum experiment, with the striking advantage of addressing precisely the noise dynamics that is affecting the shadows.
Please refer to the \hyperlink{protocol}{Protocol} in this letter; more details are presented in \cref{sec:robust_CNOT_dihedral_scheme}.

\begin{figure}[htpb]
\normalsize 
\begin{protocol}{\hypertarget{protocol}{Noise-mitigated classical shadow estimation}}

    \\[3pt]
    
    \textbf{Assumption:} Gate-independent noise channel $\Lambda$ for the Clifford and CNOT-dihedral gate-sets.
    
	\smallskip 

    \textbf{Setup:} Experimentally implemented Clifford gates $\{\phi (g)\}_g$.

    \smallskip 

    \textbf{Input:} Unknown quantum state $\rho$.

	\smallskip

	\textbf{Output:} Noise-mitigated classical shadow $\hat \rho$.
	
	\smallskip
	
	\textbf{Procedure:}
 
    \begin{enumerate}[label=(\roman*)]
        \item Fix a CNOT-dihedral sequence length $m$.
        \item Pick a single random Clifford gate $c$ and $m$ gates from the CNOT-dihedral group, $k_1, \dots, k_m$. 
        \item Implement the sequence $\phi(k_m) \dots \phi(k_1) \phi(c)$ on the target state $\rho$. 
        \item Measure in the computational basis, obtain outcome~$b'$. 
        \item Classically compute $(d+1)\Tr[E \, \Pi_\adj \Ad^\dagger(g_\eend )(b')]$, with $g_\eend= k_m\cdots k_1 \cdot c$.
        \item Repeat the above procedure $N=\mc O(1)$ times to estimate the sample average $\mc Q (m+1)$.
        \item Increase the sequence length $m+1$ and repeat the above procedure.
        \item Fit the data into the model in \cref{shadow_dihedral_fitting_model} to extract~$\lambda_Z(\Lambda)$.
        \item Compute the calibrated the frame operator $\hatM_{m+1}$ according to \cref{eq:RB_calibrated_frame_circuit}.
        \item Produce the noise-mitigated classical shadows $\hat \rho_{m+1}^{(j)} = \hatM^{-1}_{m+1}\Ad^\dagger(g_\eend) \kett{b'}$ for $j=1,\dots,N$.
    \end{enumerate}
\end{protocol}
\end{figure}

A simpler scheme results when using classical shadows and conventional RB both with uniformly drawn multi-qubit Clifford gates. 
We can again obtain calibration and estimation from the same experimental data. 
For this to work, we require the additional assumption
\begin{equation}\label{approx:lambda_vs_F}
    \lambda_\adj\approx \lambda_Z,
\end{equation}
where $\lambda_\adj$ is the parameter obtained from the conventional Clifford RB method, often related to the effective depolarizing strength. 
We can use this figure to estimate the frame operator by classically computing
\begin{equation}\label{MhatClifford}
    \hatM_\adj	= \Pi + \frac{\lambda_\adj}{d+1} \Pi_\adj.
\end{equation}
Remarkably, the approximation in \cref{approx:lambda_vs_F} turns out to be accurate for a number of noise models. 
In \cref{sec:global_clifford} we analytically derive the explicit ratio of the decay parameters for common noise models -- depolarizing noise, bit-flip noise, and amplitude damping -- and demonstrate numerically that in these scenarios the error bias in \cref{eq:calibrated_bias} is significantly smaller than the uncalibrated bias in \cref{eq:uncalibrated_bias}.

Finally, let us study a setting where we want to implement a protocol involving \emph{local} Clifford gates only, considerably reducing the experimental requirements. 
In this case, it is common to consider the task of learning expectation values of observables of the form $O = \sum_i O_i$, where each $O_i$ has non-trivial support of bounded locality. 
In this case, it suffices to produce a calibration of the frame operator restricted to the local supports. This is given by 
\begin{equation}\label{eq:local_frame_M} 
    \hatM_\loc(\{O_i\}_i) \coloneqq \sum_{w\  :\ \supp{w} \subseteq \ \bigcup\supp{O_i}} p_{w,B}  \, \Pi_w,
\end{equation}
where $\Pi_w$ are the projectors onto the subset of operators with vanishing local partial trace on the subsystems in the support of $w \in \{0,1\}^n$, cfr.~\cite{RobustShadows}.

The required calibration parameters $p_{w,B}$ can be related to notions of an effective local dephasing strength. 
Ref.~\cite{RandomSequences} studies extensions of RB protocols that allow for interleaving so-called probe operators in the post-processing, thereby constructing classical shadows of the average noise associated to a gate set. 
Here, we use this gate-set shadow protocol with sequences of parallel local random Clifford unitaries to estimate the parameters $p_{w,B}$. 
Roughly speaking, we interleave filter functions with the operator $B$ in the post-processing to estimate local dephasing strengths instead of local depolarizing strengths. 
The complexity bounds provided in Ref.\  \cite{RandomSequences} ensure efficient learning of $p_{w,B}$ of size $\mc O (\log n)$, yielding an efficient estimation of $\hatM_\loc(\{O_i\}_i)$ in \cref{eq:local_frame_M} under gate-independent noise and observables with $\mc O (\log n)$-bounded support. 
We refer to \cref{sec:local_Clifford} for details on  the protocol.

\begin{figure}
  \centering
\includegraphics[width=0.38\textwidth]{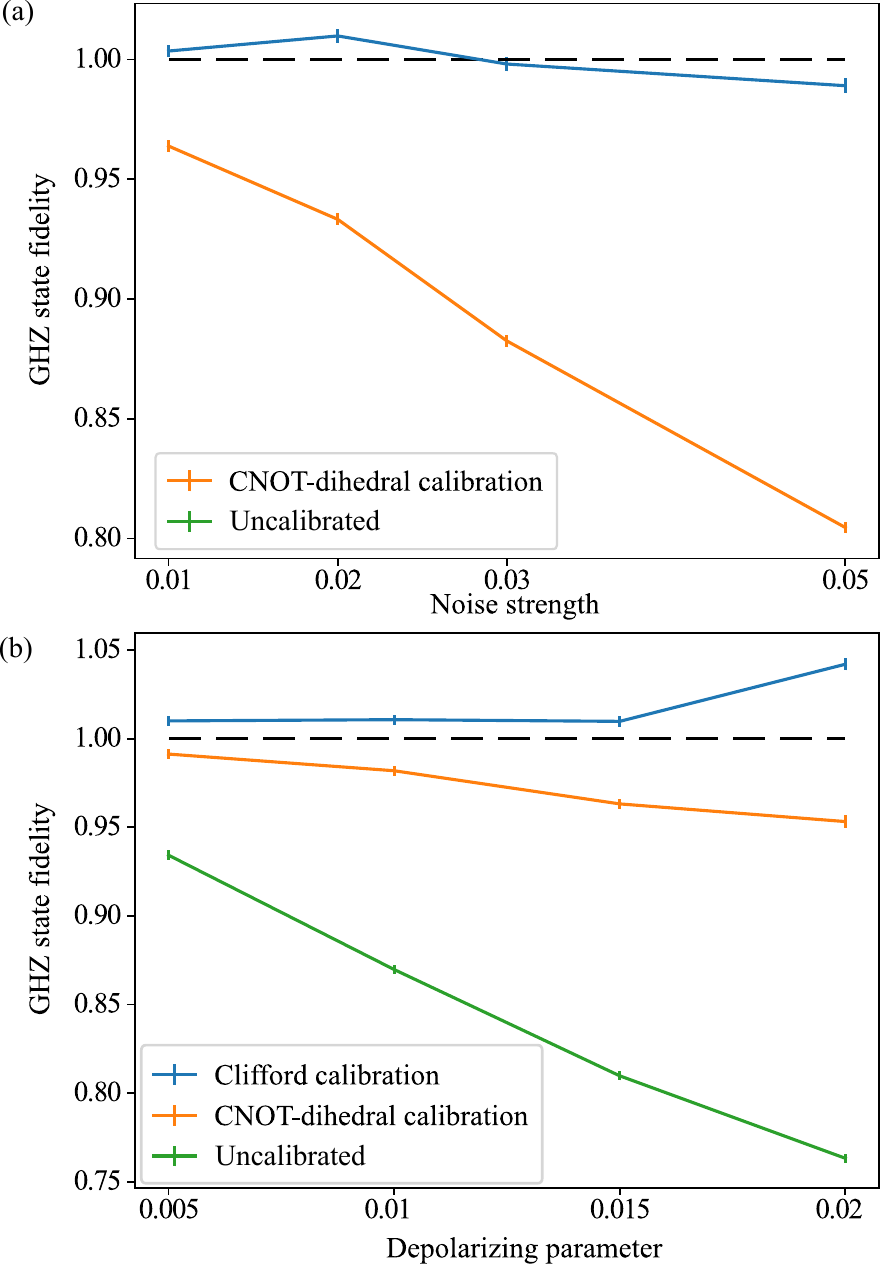}
  \caption{(a) 
  Application of the RB-calibrated  shadow estimation scheme for gate-independent noise applied to 8 qubits, for amplitude damping and dephasing, each with parameter $p\in [0,1]$ shown on the horizontal axis. (b) The performance of 
  the scheme for gate-dependent noise applied on 5 qubits. We consider a simple gate-dependent noise model in which CNOT gates are followed by two-qubit depolarizing noise with parameter $p$ shown on the horizontal axis. In both cases, the task is to estimate the state fidelity with a GHZ state. The noise-free target value is shown by the black dashed line. Error bars show one bootstrap standard deviation of the median-of-means estimator. See \cref{sec:numerics_details} for details on the numerical simulation.}
  \label{fig:main_numerics}
\end{figure}

\paragraph*{Performance guarantees.} 

Assuming a gate-independent noise channel $\Lambda$ acting after the implementation of the ideal unitary, we model the noisy implementation $\phi(g)$\footnote{Note that the way we make use of this assumption in our work, this noise model does not include read-out noise (contrary to the similar assumption made in Ref.~\cite{RobustShadows}). 
We might assume that the read-out noise was already addressed by a separate read-out mitigation protocol \cite{Maciejewski2020mitigationofreadout, Karalekas:2020:Quantumclassical,vanDerBerg:2022:Model-free}.
} of an ideal gate $\Ad(g)$ as
\begin{equation}\label{eq:noise_model}
	\phi(g) = \Lambda \,  \Ad(g) .
\end{equation}
Note that this noise model is even more restrictive as the gate-independent noise model often used in guarantees for RB, as we do not allow for noise to act before the gate.
Then, the \emph{physical frame operator} (that is, the frame operator accounting for the noisy implementation of the quantum unitaries) corresponds (as per Ref.~\cite{RobustShadows}) to
\begin{equation}\label{eq:natM}
	\natM \coloneqq \Pi_\triv + \frac{\lambda_Z(\Lambda)}{d+1} \Pi_\adj,
\end{equation}
where the $Z$-parameter~$\lambda_Z(\Lambda)$ for gate-independent noise is
\begin{equation}\label{eq:Z-par_independent_noise}
    \lambda_Z(\Lambda) 
    =
    \frac{\Tr[\Pi_\adj B \Lambda]}{\Tr[\Pi_\adj B]}\\
    =
    \frac{\Tr[\Pi_\adj B \Lambda]}{d-1}.
\end{equation}
An uncalibrated shadow estimation procedure thus yields a state estimator with a bias of the form
\begin{equation}\label{eq:uncalibrated_bias}
 \mc M_{\mathrm{ideal}}^{-1} \, \natM -\1 
 = 
 \left(\lambda_Z(\Lambda) - 1\right) \Pi_\adj .
\end{equation} 
To mitigate this bias in the estimation, it is thus necessary to access the quantity $\lambda_Z(\Lambda)$.
Strikingly, one of the decay parameters from an RB scheme drawing from the CNOT-dihedral group (plus the optional first gate from the Clifford group) coincides \emph{exactly} with the quantity we need to extract to reconstruct the frame operator as per \cref{eq:natM}. Under the gate-independence assumption, 
we can hence reconstruct the physical frame operator~$\natM$ in full (and obtain its inversion for the shadow protocol).

The global Clifford scheme works under the assumption of \cref{approx:lambda_vs_F}, which for gate-independent noise translates to
\begin{equation}
    \lambda_\adj(\Lambda)=\frac{\Tr[\Pi_\adj \Lambda]}{d^2-1} \approx  \frac{\Tr[\Pi_\adj B \Lambda]}{d-1}= \lambda_Z(\Lambda),
\end{equation}
yielding an error bias 
\begin{equation}\label{eq:calibrated_bias}
	\hatM_\adj^{-1}  \, \natM - \1 
	=
    \left\{\frac{\lambda_Z(\Lambda)}{\lambda_\adj (\Lambda)}  - 1 \right\} \Pi_\adj .
\end{equation}


It is key to the approach suggested in this work that not only
an error bias is mitigated, but also that the scalability of the sample complexity is precisely guaranteed. In \cref{app:variance} we include error bounds on the variance of the derived variants of robust shadows estimation. In particular, in \cref{app:variance_dihedral} we demonstrate that the variance of the CNOT-dihedral scheme is~$\mc O (1)$ in the number of qubits and sequence lengths.

\paragraph*{Numerical examples.} 
We numerically simulated the protocol for the task of estimating the fidelity of a GHZ-state where the state preparation is ideal for various noise models affecting the gates randomizing the measurement. 
In this way, every deviation of the fidelity away from 1 is caused by noise and errors in the randomized measurement. 
A selection of numerical results is shown in \cref{fig:main_numerics},  further results are presented in \cref{app:noise_examples} and \cref{app:gate_dependence}. 
We find that the separate calibration with CNOT-dihedral RB can successfully correct the frame operator for arbitrary gate-independent noise. 
The performance of the (global) Clifford group calibration depends on the type of (gate-independent) noise. 
In accordance with our theoretical derivations, certain noise models still have a bias in the estimation.  
Nonetheless, we consistently observe improvements over uncalibrated estimation for the parameter regime under consideration. 
For gate-dependent noise, we observe  that our method improves the accuracy of shadow estimation for all considered noise models, albeit leaving a small bias in all cases.

\paragraph*{Conclusions and outlook.} 
In this work, 
we have developed methods for mitigating the noise in randomized measurement schemes that is introduced by the gates randomizing the measurements. 
We argued that the flexible toolkit provided by randomized benchmarking  methods and its extension to gate-set shadows 
can be successfully put to use in order to characterize the relevant parameters of the noise affecting randomized measurements. 
As a result, using these calibration parameters, noise in the randomized measurements can be mitigated or even completely accounted for in classical post-processing. 
We have exemplified different variants of the approach for randomization with the multi-qubit  Clifford group, assisted by either standard RB or filtered RB using the CNOT-dihedral group, and with local Clifford unitaries. 
In particular, with slight modifications the RB protocol for calibration and the estimation protocol can use the same measurement data, giving rise to self-calibrating shadow estimation schemes.
The performance of the scheme can be understood using a gate-independent noise assumption, and we provide rigorous guarantees in this setting. 
This is complemented by numerical demonstrations of the protocol's performance for practically motivated noise examples. 
Besides its efficiency we inherit the robustness against state-preparation from RB.

Possibly the most important conceptual take-away message from this work is of highly practical nature: 
sometimes the available set of techniques and data might already be 
sufficient to mitigate errors, without the need to resort to more sophisticated characterization experiments.  
Instead, much of the burden can be placed on the theoretical understanding of the relation between the parameters entering  different protocols under plausible noise assumptions. 
We here take a first step to draw such connections, still using arguably restrictive assumptions on the noise models. 
We hope that this motivates further work on suitable approximations of average noise characterization under a more flexible set of assumptions. 

\paragraph*{Acknowledgements.} We would like to thank Richard Kueng, Markus Heinrich, Leandro Aolita and Kristan Temme for fruitful discussions. EO~has been supported by the ERC under grant agreement no.~101001976 (project EQUIPTNT).
The Berlin team has been supported by the BMBF (DAQC, MUNIQC-Atoms), the Munich Quantum Valley (K-4 and K-8), the Quantum Flagship (PasQuans2, Millenion), QuantERA (HQCC), the Cluster of Excellence MATH+, the DFG (CRC183), the Einstein Foundation (Einstein Research Unit on Quantum Devices), and the ERC (DebuQC). JH acknowledges support from the Quantum Software Consortium (NWO Zwaartekracht 024.003.037) and a Veni grant (NWO Veni 222.331). AHW thanks the VILLUM FONDEN for its support with a Villum Young Investigator (Plus) Grant (Grant No.\ 25452 and Grant No. 60842) as well as via the QMATH Centre of Excellence (Grant No.\ 10059).

\paragraph*{Note added.} Upon completion of this work, we became aware of the manuscript  \cite{Harvard} that is related to the goal of our work for error-mitigated shadow estimation.

\let\oldaddcontentsline\addcontentsline
\renewcommand{\addcontentsline}[3]{}

\let\addcontentsline\oldaddcontentsline

\cleardoublepage
\onecolumngrid

\largesection
\largesubsection
\largesubsubsection

\appendix
\renewcommand{\appendixpagename}{\centering \Large Supplementary material for noise-mitigated randomized measurements}
\appendixpage

\vskip30pt

\tableofcontents

\vskip30pt

\section{Summary and guidance}

This supplementary material has been set up to fulfill two purposes. 
On the one hand, it provides further detail concerning the technical and mathematical statements on the main text.

On the other hand, it provides further guidance for practitioners in experimental physics who wish to apply our results in the laboratory, or to theorists supporting such experiments. On the highest level, such practitioners should ask themselves three questions.

\begin{itemize}
\item Would they wish to pursue global or local Clifford shadow estimation?

\item What is the presumed prevalent noise in the device?

\item Do they want to calibrate and estimate from the same data, or do it separately?

\end{itemize}

Following this set of options, we present four calibration methods that differ both in their setup and features, but all of them are built on random circuits and data fitting. For all of them the gate-independent noise assumption is introduced, and we denoted the error channel by~$\Lambda$. More precisely, we can think of the following four specific uses of our scheme of the main text.

\begin{table}[htbp]
\centering\setlength\tabcolsep{3.5pt}\renewcommand\arraystretch{2.25}
  \noindent\makebox[\textwidth]{%
    \begin{tabular}{|l|*{4}{c|}}
      \hline
      \diagbox[width=\dimexpr \textwidth/8+4\tabcolsep\relax, height=1.35cm]{Scheme}{Features}
                   & Exact calibration & Self-calibrating & Post-processing & Local gates \\
      \hline
      CNOT-dihedral & Yes & No & None & No\\
      \hline
       \makecell{CNOT-dihedral\\ + one Clifford} & Yes & Yes & Yes (easy) & No \\
      \hline
      Global Clifford & No & Yes & None & No\\
      \hline
      Local Clifford & Yes & No & Yes & Yes\\
      \hline
    \end{tabular}
  }%
  \caption{Features and trade-offs of the four specific calibration schemes.}
\end{table}

\begin{enumerate}[label=(\arabic*)]

    \item \textbf{Native CNOT-dihedral RB}. This is a straightforward RB procedure with gates drawn from the CNOT-dihedral subgroup only. The data are fit in a fitting model with one decay parameter, with no additional post-processing. The decay parameter is the $Z$-parameter $\lambda_Z(\Lambda)$ of the noise channel $\Lambda$. Since we are drawing gates from a subgroup only, one cannot produce classical shadows with this data. Hence, this serves as a \emph{separate} calibration procedure for shadow estimation. This scheme is presented in \cref{app:RB_dihedral}.
    
    \item \textbf{CNOT-dihedral with one random general Clifford}. This is a slight modification of the previous scheme, where one additional Clifford gate $v$ is selected and compiled together with the first dihedral gate to produce another Clifford element. Depending on the outcome $b'$ of the measurement in the  computational basis, we compute the element $(2^n+1)\braa{O}\Ad(c)\kett{b'}$ for some observables $O$ and fit the sample average into a fitting model whose only decay parameter is again the $Z$-parameter $\lambda_Z(\Lambda)$. At the cost of a simple post-processing action, we have now the benefit of being able to construct \emph{both} classical shadows and retrieve the $Z$-parameter for calibration. This variant of the scheme is illustrated in \cref{sec:robust_CNOT_dihedral_scheme}.

    \item \textbf{Conventional global Clifford RB}.
    By carrying out the conventional shadow estimation scheme with global Cliffords, one can use the measurement data -- without any specific post-processing --  to obtain a quantity which does not correspond to the $Z$-parameter, but still constitutes in many scenarios a close approximation to it. For the benefit of a very simple scheme that allows to calibrate the frame operator with the conventional shadow estimation procedure, here the additional requirement is that some structure of the noise model must be known in order to establish a-priori whether this approximation scheme delivers a reliable calibration. 
    A discussion regarding the noise dynamics this scheme is suitable for, together with the description of this variant of robust shadow estimation, is to be found in \cref{sec:global_clifford}.

    \item \textbf{Local Clifford shadow estimation scheme}. 
    Circuits made of local Clifford gates only can be used to calibrate exactly the frame operator, and obtain a robust version of shadow estimation in one go.
    Here a more sophisticated post-processing of the data is required, following the method of~\cite{RandomSequences}. Indeed, this scheme permits to learn with the same data exponentially many ``local $Z$-fidelities'' at once, constrained by the size of the support which must be of size $\mc O (\log n)$ to guarantee scalability. 
    Full details in \cref{sec:local_Clifford}.

\end{enumerate}

\FloatBarrier

\section{Preliminaries}

\subsection{Notation}\label{app:notation}
We denote by the $d^2$-dimensional vector $\kett{\rho}$ the 
\emph{vector representation} of a density operator $\rho$ acting on a $d$-dimensional subspace, with entries 
\begin{equation}
\kett{\rho}_{j,k}=\langle e_j, e_k \kett{\rho}
:=\sandwich{e_j}{\rho}{e_k},
\end{equation}
where $\{\ket{e_j}\}_j$ is the natural basis of $\mathbb{C}^d$.
Analogously, we represent representation of a quantum channel~$\Lambda$ by a $d^2 \times d^2$ matrix whose entries are defined by
\begin{equation}
\Lambda_{(j,k),(\ell,m)} \sandwich{e_j,e_k}{\Lambda}{e_\ell,e_m}\coloneqq \Tr[\ketbra{e_k}{e_j} \Lambda(\ketbra{e_\ell}{e_m})].
\end{equation}
Under this isomorphism, we then have $\kett{\Lambda(\rho)}= \Lambda\kett{\rho}$. The \emph{adjoint action} will be denoted by $\Ad (g)$ throughout this work.

In the following, we will consider two relevant groups in quantum computing, that is, the Clifford group that we denote by $\mc C$ and the CNOT-dihedral subgroup $\KK = \langle \mathrm{CNOT}, S, X \rangle$~\cite{CNOTDihedral}, with $S$ being the phase gate. 
We model the physical, possibly noisy, implementation $\phi$ of a quantum gate $g$ from the Clifford group according to an error channel $\Lambda$ that acts right after the ideal action of the gate. In matrix representation, we write 
\begin{equation}\label{eq:channel_implementation}
\phi(g) = \Lambda(g) \Ad(g). 
\end{equation}

This will grant the results on error bias and variance we derive in the following.
An element of the computational basis will be denoted as $b$, and under the vectorization isomorphism we well write a computational basis measurement by $\braa{b}$ and the matrix $\ketbra{b}{b}$ as $\kett{b}$.
In our matrix representation formalism, we also define the matrix
\begin{equation}\label{B}
	B\coloneqq \sum_b \kettbraa{b}{b},
\end{equation}
physically corresponding to the sum of measurements in basis $b$ followed by the reconstruction of the basis operator.

According to Schur's lemma, taking the group average of the object $\Ad^\dagger (g) M \Ad(g)$ for any matrix $M$ produces a ``projection onto $\Pi_\alpha$-projectors''  for the invariant sub-spaces of the irreducible representations (irreps)~$\alpha$ appearing in the decomposition of the $\Ad$ representation. Formally, we have~\cite{gambetta2012characterization}
\begin{equation}
    \EE_\GG \Ad^\dagger (g) M \Ad(g) = \sum_{\alpha\ \text{in} \Ad} \frac{\Tr[\Pi_\alpha M]}{\Tr[\Pi_\alpha]}\Pi_\alpha .
\end{equation}

\textbf{Assumption:} In this appendix, we will assume that the noise channel in \cref{eq:channel_implementation} is not dependent on the gate, hence writing $\Lambda(g) \equiv \Lambda$. Additionally, we assume that $\Lambda$ is trace preserving. 

\subsection{Randomized benchmarking}\label{app:RB}

\emph{Randomized benchmarking} (RB) 
~\cite{KnillBenchmarking,MagesanPRL,frameworkRB}
refers to a large family of noise characterization schemes with similar characteristics. The common trait of all these procedures is that they make use of sequences of random quantum gates followed by a POVM measurement.
With the collected data, we compute a \emph{sample average} for varying sequence lengths, which are inserted into a theoretical model to extrapolate and estimate a set of parameters. These serve as figures of merit to quantitatively assess the level of noise which is affecting the experimental realization of the gate-set.

The key point of RB is that a scalable number of samples suffices to closely approximate the theoretical average~$\mc Q_\avg$ over all possible random samples. The fitting model for~$\mc Q_\avg$ has the form of linear combination of matrix exponential decays~\cite{frameworkRB}, each of them connected to a different irreducible representation in the decomposition of the $\Ad$ representation of the group from which the random sequences are drawn.
In the particular case where no irreducible representation appears more than once in the decomposition, that is, the $\Ad$ representation is multiplicity-free, the fitting model simplifies into a sum of scalar parameter decays~$\lambda^m$, one for each irreducible representation $\alpha$ in the decomposition. Namely,
\begin{equation}\label{eq:RB_fitting_model}
    \mc Q_\avg(m) = \sum_{\alpha} y_\alpha(E, \rho) \lambda_\alpha^m,
\end{equation}
where $\{y_\alpha\}_\alpha$ are constants depending on $m$, the input state~$\rho$, the POVM measurement~$E$ and the end gate $g_\eend$ resulting from the whole circuit. This gate can either be chosen to be a specific operator, or selected at random (which is a central requirement for shadow estimation). 
By absorbing the measurement errors in the preparation of $\rho$ and realization of $E$, they make RB protocols robust against this kind of errors.
A second highly favorable feature of RB is the fact that they are resilient against  gate-dependent noise~\cite{IndependentNoise,frameworkRB}. This is a key advantage compared to other calibration schemes that are less stable under gate-dependent errors~\cite{brieger2023StabilityClassicalShadows}.

Among the many RB variants, for this work the relevant ones are the \emph{uniform} type, where one considers a group as the gate set from which gates are drawn uniformly at random.  That is, for an input state $\rho$ and a given sequence length $m$, we produce measurements with respect to some POVM~$E$ of a circuit generated by $m$ random gates and a final gate~$g_\eend$,
\begin{equation}
    \mc Q (m,g_\eend,\rho,E) 
    \coloneqq 
    \braa{E}\phi(g_\eend g_\inv) \phi(g_m) \dotsb \phi(g_2) \phi(g_1) \kett{\rho},
\end{equation}
with $g_\inv := (g_m \cdots g_2 g_1)^{-1}$.
The resulting measurement data $Q (m,g_\eend,\rho,E)$ are then used in the post processing stage -- possibly together with some other classically computed filtering functions -- to generate the sample average~$\mc Q_\avg(m)$ for the fitting model in \cref{eq:RB_fitting_model}.
Regarding the choice of the gate set, for our calibration routine we will employ uniform RB procedures drawing from the Clifford group~\cite{KnillBenchmarking,MagesanPRL} and the CNOT-dihedral group~\cite{CNOTDihedral}.

\subsection{Shadow estimation}\label{app:frame}

Given an unknown quantum state we are interested in characterizing, the shadow estimation scheme~\cite{Shadows} allows for the estimation of many observables from a so-called classical shadows of the state, i.e.\ measurement data from an IC-POVM implemented by measuring in random orthonormal bases. 
In specific settings, e.g. fidelity estimation with multi-qubit Clifford rotations, the number of samples (or copies of the state) required scales efficiently in the system size and the estimation precision, and logarithmically with respect to the number of observables. 

\medskip 

To carry out the procedure, we first randomly select an element $g \in \mc C$, and apply the physical quantum channel $\phi(g)$ on $\rho$. Then, we measure in the basis $\{b\}_b$ obtaining an outcome $b'$. 
Finally, we apply the inverse of the frame operator to the classically and efficiently stored object $\Ad^\dagger (g) \kett{b'}$. Formally, a \emph{classical shadow} of the quantum state $\rho$ is constructed by
\begin{equation}
	\kett{\hat \rho \, (g,b')} =  \mc M_\mathrm{ideal}^{-1}  \Ad^\dagger (g)\kett{b'},
\end{equation} 
where the frame operator $\mc M_\mathrm{ideal}$ is defined and calculated in Ref.~\cite{scott2006tight} as
\begin{equation}\label{eq:uncalibrated_frame}
    \mc M_\mathrm{ideal}
    \coloneqq 
    \EE_{\mc C} \Ad^\dagger (g) B \, \Ad(g)
    =
    \left( \frac{\Tr[\Pi_\triv B]}{\Tr[\Pi_\triv]} \Pi_\triv + \frac{\Tr[(\Pi_\adj B)]}{\Tr[\Pi_\adj]} \Pi_\adj \right) \\
    =
    \Pi_\triv + \frac{1}{d+1} \Pi_\adj,
\end{equation}
with 
\begin{equation}
\Pi_\triv:= {d}^{-1} \sum_{j,k=1}^d \kettbraa{j}{k}
\end{equation}
being the vectorization representation of the rank-1 projection onto the identity operator connected to the trivial invariant subspace, and consequently $\Pi_\adj = \1_{d^2}- \Pi_\triv$ the projection onto its complement subspace related to the adjoint irreducible representation. 

\medskip

From $N_1$ sets of $N_2$ different classical shadows 
$\hat \rho_1, \dots\hat \rho_{N_2}, \hat \rho_{N_2+1} ,\dots, \hat \rho_{2N_2},\dots \hat \rho_{N_1\cdot N_2}$ 
one first constructs $N_1$ objects 
\begin{equation}
   o_i (k) = \frac{1}{N_2} \sum_{j=1}^{N_2} \braakett{O_i} {\hat \rho_{j,k}}, \qquad k=1,\dots, N_1,
\end{equation}
for every observable $O_i$ in a set $\{O_i\}_{i=1}^M$.
One than constructs a median-of-means,
\begin{equation}
    \hat o_i (N_1,N_2)= \mathrm{median} \left\{ o_i(1), \dots, o_i(N_1) \right\},
\end{equation}
so that one can formulate the following precise statement~\cite{Shadows,lugosi2019mean}.
Set 
\begin{equation} 
    N_1= 2 \log(2M/\delta) \quad \text{and} \quad N_2 = \frac{34}{\varepsilon^2} \max_i \norm{O_i-\braakett{\1}{O_i} \frac{\1}{2^n} }_\mathrm{shadow}^2. 
    \end{equation} 
    Then a collection of $N= N_1 \cdot N_2$ independent classical shadows allow for accurately predicting all features via median-of-means prediction,
    \begin{equation}
        \abs{\hat o_i(N_1,N_2), - \braakett{O_i}{\rho} } \leq \varepsilon \quad \text{for all} \quad i=1, \dots, M,
\end{equation}
with probability at least $1-\delta$.
In the norm $\norm{\, \cdot \, }_\mathrm{shadow}$ denotes the shadow norm defined in Ref.~\cite{Shadows}.
The most significant limitation of the scheme is the shadow norm of the observables.

\medskip 

Now back to the reconstruction of the classical shadows, we note that the \emph{frame operator} is a central object for the scheme. 
Formally, taking into account a noisy physical implementation of the quantum operators, the frame operator is given by
\begin{equation}\label{eq:real_M}
	\mc M \coloneqq \EE_\GG \,  \Ad^\dagger(g) B \, \phi (g) ,
\end{equation}
where the unitary group $\GG$ generating the IC-POVM from the computational basis measurement is usually the Clifford group $\CC$. 
The original description of  shadow procedure estimation makes use of analytical expressions for the ideal frame operator $\mc M_\mathrm{ideal}$ with $\Phi =\Ad$, \cref{eq:uncalibrated_frame}. 
It does not account for noise and errors in the gate implementation. 

This misalignment will inevitably lead to an bias in the state shadow estimator. 
Under the assumption of a gate-independent noise channel $\Lambda$ acting after the ideal unitary, and the trace-preservation assumption implying $\Tr[\Pi_\triv \Lambda]=1$, the physical frame operator in \cref{eq:real_M} becomes~\cite{RobustShadows}
\begin{aligneq}\label{eq_app:natM}
	\natM 
	&\coloneqq 
	\EE_\CC \Ad^\dagger (g) B \, \phi(g)
	= 
    \left( \frac{\Tr[\Pi_\triv B\Lambda]}{\Tr[\Pi_\triv]} \Pi_\triv + \frac{\Tr[\Pi_\adj B \Lambda]}{\Tr[\Pi_\adj]} \Pi_\adj \right) \\
	&=
	\Pi_\triv + \frac{\lambda_Z(\Lambda)}{d+1} \Pi_\adj ,
\end{aligneq}
where we have introduced the \emph{$Z$-parameter} $\lambda_Z(\Lambda)$ \footnote{The $Z$-parameter $\lambda_Z(\Lambda)$ for a gate-independent noise channel $\Lambda$ is directly related to the $Z$-\emph{fidelity} $F_Z$ defined in \cite{RobustShadows} via $\lambda_Z(\Lambda) = (dF_Z(\Lambda) - 1)/(d-1)$. It is also proportional to the depolarizing parameter $f(\Lambda)$ used in \cite{koh2022ClassicalShadowsNoise} via $\lambda_Z(\Lambda) = (d+1)f(\Lambda)$.} of the noise channel $\Lambda$ as
\begin{equation}\label{eq:lambda_Z}
    \lambda_Z(\Lambda) \coloneqq \frac{\Tr[\Pi_Z \Lambda]}{\Tr[\Pi_Z]} 
    = 
    \frac{\Tr[\Pi_Z \Lambda]}{d-1} ,
\end{equation}
with $\Pi_Z$ being the projector onto Pauli strings made of identities and $Z$ operators only, that is, 
\begin{equation}\label{def:Pi_Z}
\Pi_Z \coloneqq B - \Pi_\triv = \Pi_\adj B = \sum_{p\in \{\1,Z\}^{\otimes n} \backslash \1} \kett{p/\sqrt{2^n}}\braa{p/\sqrt{2^n}} .
\end{equation} 
The $Z$-parameter will play a central role throughout the following constructions. 
Using \cref{eq_app:natM}, we can quantify the error bias of a shadow estimation procedure with an uncalibrated inverse according to the following expression.

\begin{lemma}[Error bias for the uncalibrated frame operator~\cite{RobustShadows}]\label{lemma:diff_uncalibrated_inv}
\begin{equation}
	\mc M_\mathrm{ideal}^{-1} \, \natM -\1 
	=
	\Big[\Pi_\triv + (d+1)\Pi_\adj\Big] 
	\Big[\Pi_\triv + \frac{\lambda_Z(\Lambda)}{d+1} \Pi_\adj \Big] - \1 \\ 
	=
	\left( \lambda_Z(\Lambda) - 1 \right) \Pi_\adj.
\end{equation}
\end{lemma} 
\noindent Looking at \cref{lemma:diff_uncalibrated_inv}, we observe that the $Z$-parameter $\lambda_Z$ is the element distinguishing the calibrated frame operator $\hatM$ from its uncalibrated version. 
Note that if $\lambda_Z(\Lambda)$ is (approximately) equal to $1$, the shadow estimation scheme does not need any sort of calibration. 

\bigskip

In the following \cref{sec:CNOT-dihedral}, we will present two calibration protocols to retrieve $\lambda_Z(\Lambda)$ through an RB procedure and construct our estimator of the frame operator as
\begin{equation}\label{eq:RB_calibrated_frame}
    \hatM 
    \coloneqq 
    \Pi_\triv + \frac{\lambda_Z (\Lambda)}{d+1} \Pi_\adj.
\end{equation}
The two strategies differ by the first gate of the random sequence. Indeed, in \cref{sec:robust_CNOT_dihedral_scheme} we will illustrate how adding a single random global Clifford gate at the beginning of the sequence and by slightly modifying the post-processing phase we can construct shadows of an unknown quantum state \emph{and} calibrate the frame operator in one go.

\section{CNOT-dihedral group calibration}\label{sec:CNOT-dihedral}

Here the assumption is that the  gate-independent noise channel of the Clifford group is the same as the gate-independent noise of the CNOT-dihedral subgroup. 
This assumption is more subtle than the standard gate-independent noise assumption in the following sense. The gate-independent noise channel adopted to describe the gate implementation of a gate set can be considered as the noise averaged over the set. Requiring that the average noise of a subgroup correspond to the average noise of the entire group is hence a stricter assumption than the standard gate independent noise assumption. In practical sense however we can match the noise for the CNOT-dihedral subgroup to the one of the whole Clifford group by compiling the CNOT-dihedral gates using some Clifford gates.

\medskip

The $\Ad$ representation of the CNOT-dihedral group splits into three irreducible representations: (i)~the trivial representation, (ii)~the irreducible projection carrying the invariant subspace of identities and $Z$ operators, whose projector is $\Pi_Z$ as per \cref{def:Pi_Z},
and 
(iii)~the orthogonal complement, given by $\Pi_\ort := \1_{4^n} - \Pi_{\triv} - \Pi_{Z}$. 
The representation we are interested in is the one carrying the Pauli-$Z$ operators, since the projection of the noise channel $\Lambda$ onto $\Pi_Z$ defines $\lambda_Z(\Lambda)$ as per \cref{eq:lambda_Z}.

\subsection{\textit{Z}-parameter extraction with randomized benchmarking}\label{app:RB_dihedral}

The protocol involves the construction of random circuits of different lengths sampling from the CNOT-dihedral subgroup to be applied on the initial state $\rho_0$.  
Dihedral RB for single qubits has been introduced in Ref.~\cite{DugWallEme15} and extended to muti-qubits as CNOT-dihedral RB in Ref.~\cite{CNOTDihedral}.  
We will take a slightly different approach than Ref.~\cite{CNOTDihedral} and use filter functions \cite{frameworkRB} to isolate the decay parameter associated to the different irreps of the generated group. 
This procedure will return the exact value for $\lambda_Z(\Lambda)$ that we can use to calibrate the frame operator.
Re recall that the CNOT-dihedral group is generated by $\mathbb{K}= \langle \mathrm{CNOT}, S, X \rangle$, with $S$ being the phase gate.
We choose $\rho_0=\ketbra{0}{0}^{\otimes n}$ such that $\Pi_\ort\kett{\rho_0}=0$. Hence, the RB fitting model is (already derived in the general form with two decay parameters in Ref.~\cite{CNOTDihedral})

\begin{lemma}[Fitting model for CNOT-dihedral RB]
The fitting model for the CNOT-dihedral RB scheme takes the approximate form $\mc Q(m) \approx a_0 \, (\lambda_Z)^m +b_0$, with $\lambda_Z(\Lambda)$ as defined in \cref{eq:lambda_Z} and where $a_0$ and $b_0$ constants absorbing SPAM errors.
\end{lemma}

\begin{proof}
We have
\begin{equation}
\begin{aligned}
    \mc Q_\avg (m) &= \EE_{k_1,\dots k_m \in \KK} \braa{\widetilde E} \phi(k_\inv) \phi(k_m) \dots \phi(k_2) \phi(k_1) \kett{\widetilde \rho_0} \\
    &= \EE_{h_1,\dots h_m \in\KK} \braa{\widetilde E} \Lambda \Ad(k_\inv) \Lambda \Ad(k_m) \dots \Lambda \Ad(k_2) \Lambda \Ad(k_1) \kett{\widetilde \rho_0} \\
    &= \EE_{k'_1,\dots k'_m \in\KK} \braa{\widetilde E} \Lambda \Ad(k'_m)^\dagger \Lambda \Ad(k'_m) \, \Ad(k'_{m-1})^\dagger \Lambda \Ad(k'_{m-1})\dots \Ad(k'_1)^\dagger \Lambda \Ad(k'_1) \kett{\widetilde \rho_0}\\
    &= \braa{\widetilde E} \Lambda \Big(\EE_{k' \in \KK} \Ad(k')^\dagger \Lambda \Ad(k') \Big)^m\kett{\widetilde \rho_0} \\
    &= \braa{\widetilde E} \Lambda \Big( \sum_{\alpha\in \{\triv, Z, XY\}} \frac{\Tr[\Pi_\alpha \Lambda]}{\Tr[\Pi_\alpha]} \Pi_\alpha \Big)^m\kett{\widetilde \rho_0} \\
    &= \sum_{\alpha\in \{\triv, Z, \ort\}}   \braa{\widetilde E} \Lambda \Pi_\alpha \kett{\widetilde \rho_0} \, \lambda_\alpha(\Lambda)^m.
\end{aligned}
\end{equation}

\noindent Assuming that the implementations of $E$ and $\rho_0$ are not too noisy, $ \braa{\widetilde E} \Lambda \Pi_\ort \kett{\widetilde \rho_0} \approx 0$
and we can neglect the $\ort$ subspace, whereas $\lambda_\triv(\Lambda)=1$.
Alternatively, we can use filter functions to isolate the sought parameter.
This yields 
\begin{equation}
    \mc Q_\avg (m) \approx \braa{\widetilde E} \Lambda \Pi_Z \kett{\widetilde \rho_0} \, \lambda_Z(\Lambda)^m + \braa{\widetilde E} \Lambda \Pi_\triv \kett{\widetilde \rho_0},
\end{equation}
which ends the proof.
\end{proof}

Hence, under the above assumptions on the noise model and up to an imprecision due to the RB parameter assessment, we can learn the $Z$-parameter with high precision and in a scalable fashion by using a CNOT-dihedral randomized benchmarking experiment~\cite{carignan2015characterizing,CNOTDihedral} to fully reconstruct~$\natM$.

\begin{protocol}{CNOT-dihedral $\lambda_Z(\Lambda)$ extraction}

    \\[3pt]
    
    \textbf{Assumption:} Gate-independent noise channel $\Lambda$ for the CNOT-dihedral gate-set.
    
	\smallskip 
	
	\textbf{Setup:} Experimentally implemented CNOT-dihedral gates $\{\phi (k)\}=\{\Lambda\Ad(k)\}$.

	\smallskip

	\textbf{Output:} Parameter $\lambda_Z(\Lambda)$.
	
	\smallskip
	
	\textbf{Procedure:}
	 
	 \smallskip
	 
	\begin{enumerate}[label=(\roman*)]
            \item Fix a sequence length $m$.
            \item Prepare $\rho_0=\ketbra{0}{0}^{\otimes n}$.
		  \item Apply a random sequence of CNOT-dihedral gates with a final inverse gate,
		  $\phi(k_\inv)\phi(k_m) \phi(k_{m-1}) \dotsb \phi(k_1)$ to the $\rho_0$.
		  \item Measure with respect to POVM $E$.
            \item Carry out the above procedure $\mc O(1)$ many times and produce a sample average.
            \item Increase the value of $m$ and repeat the above.
		  \item Fit the sample averages for different $m$ to the exponential decay model and extract the RB decay parameter~$\lambda_Z(\Lambda)$ .
	\end{enumerate}
\end{protocol}

\subsection{Noise-mitigated shadow scheme}\label{sec:robust_CNOT_dihedral_scheme}

While the above scheme is suited to extract the $Z$-parameter, 
rotating the computational basis with the CNOT-dihedral group will not generate an IC-POVM. 
Thus, it does not produce a state shadow. 
However, by slightly modifying the scheme with the addition of a general, random Clifford gate at the beginning of each sequence, we can use  the \emph{same} measurement data  to produce both shadows of the target unknown quantum state \emph{and} the $Z$-parameter. With this procedure, we can thus guarantee that we are precisely addressing the noise affecting the classical shadows, and correct the frame operator.

\medskip 

We proceed as follows. We use our target unknown quantum state $\rho$ as input.
We produce a random circuit of length $m+1$, where the first gate $c$ is a random Clifford, by $m$ random gates $k_1,\dots,k_m$ drawn from~$\KK$. We denote the composition of all gates in the circuit by $g_\eend = k_m \cdots k_1 \cdot c$.
We then measure in the computational basis. 
From the outcome $b'$ we construct \emph{classically} the filter function 
\begin{equation}
    (d+1)\braa{E} \Pi_\adj \Ad^\dagger(g_\eend)\kett{b'}.
\end{equation}
We repeat this procedure multiple times, for varying circuit lengths~$m+1$. The sample average $\mc Q (m+1)$ will then be used to extract the $Z$-parameter according to a suitable fitting model.

\begin{lemma}[Fitting model for shadow CNOT-dihedral scheme]\label{shadow_dihedral_fitting_model}
  The theoretical average of the survival probability of the self-calibrating shadow estimation protocol obeys an exponential law in the sequence length $m$,
  \begin{equation}
    \mc Q(m+1) 
    = 
    \braa{E} \Pi_Z \kett{\rho} \, \lambda_Z(\Lambda)^{m+1} .
    \end{equation}
    
\end{lemma}
  
At the same time, the measurement outcome of a circuit of length~$m+1$ can be used to construct a calibrated classical shadows $\hat \rho_{m+1}$, that is, 
\begin{lemma}[Estimated frame operator for circuit shadows]\label{lemma:circuit_M_estimator}
    \begin{equation}
        \hat \rho_{m+1} = \hatM_{m+1}^{-1} \Ad^\dagger(g_\eend) \kett{b'}, 
    \end{equation}
    where in this instance the frame operator is calibrated with an exponential of the RB parameter, namely,
    \begin{equation}\label{eq:RB_calibrated_frame_onego}
        \hatM_{m+1} \coloneqq \Pi_\triv +  \frac{\lambda_Z(\Lambda)^{m+1}}{d+1}  \Pi_\adj .
    \end{equation} 
\end{lemma}

\begin{proof}[Proof of \cref{shadow_dihedral_fitting_model}]
    Let $k_\inv = (k_m \cdots k_1)^{-1}$ and $g_\eend = k_m \cdots k_1 \cdot c$. 
    The measured signal in expectation reads
    \begin{align}
        \mc Q_{\mathrm{avg}}(m+1) 
         &= 
        \EE_{c \in \mc C}\EE_{k_1,\dots,k_m \in \mathbb{K}} (d+1)
        \braa{E}\Pi_{\adj}\Ad(g_\eend)^\dagger B \phi(k_m) \dots \phi(k_1) \phi(c) \kett{\rho} \\
        \nonumber
        &=
        (d+1)\EE_{c \in \mc C}\braa{E}\Pi_{\adj}\Ad^\dagger(c) \Ad(k_\inv) B \Lambda \Ad(k_\eend) \Big(\EE_{k \in \KK} \Ad^\dagger(k)\Lambda \Ad(k)\Big)^{m-1} \Lambda \Ad(c) \kett{\rho} \\
        \nonumber
        &=
        (d+1)\EE_{c \in \mc C}\braa{E}\Pi_{\adj}\Ad^\dagger(c) 
        \sum_{\alpha\in \{\triv, Z, \ort \}} \frac{\Tr[B \Lambda \Pi_\alpha]}{\Tr[\Pi_\alpha]} \Pi_\alpha 
        \Big( \sum_{\alpha\in \{\triv, Z, \ort\}} \frac{\Tr[\Pi_\alpha \Lambda]}{\Tr[\Pi_\alpha]} \Pi_\alpha \Big)^{m-1} \Lambda \Ad(c) \kett{\rho} \\
         \nonumber
        &=
        (d+1)\sum_{\alpha\in \{\triv, Z, \ort \}} \ \EE_{c \in \mc C}\braa{E}\Pi_{\adj} \Ad^\dagger(c)      \Pi_\alpha  \Lambda\Ad(c)\kett{\rho}  \frac{\Tr[(\Pi_\triv + \Pi_Z)\Pi_\alpha \Lambda ]}{\Tr[\Pi_\alpha]} \, \lambda_\alpha(\Lambda)^{m-1}\\
        \nonumber
        &=
    	(d+1)\sum_{\alpha\in \{\triv, Z, XY\}} \braa{E}\Pi_{\adj} \Pi_{\adj} \kett{\rho} \frac{\Tr[\Pi_\adj \Pi_\alpha \Lambda ]}{\Tr[\Pi_\adj]} \frac{\Tr[(\Pi_\triv + \Pi_Z)\Pi_\alpha \Lambda ]}{\Tr[\Pi_\alpha]} \, \lambda_\alpha(\Lambda)^{m-1} \\\nonumber
        &+
        (d+1)\sum_{\alpha\in \{\triv, Z, XY\}} \braa{E}\Pi_{\adj} \Pi_{\triv} \kett{\rho} \frac{\Tr[\Pi_\triv \Pi_\alpha \Lambda ]}{\Tr[\Pi_\triv]} \frac{\Tr[(\Pi_\triv + \Pi_Z)\Pi_\alpha \Lambda ]}{\Tr[\Pi_\alpha]} \, \lambda_\alpha(\Lambda)^{m-1}\\ 
    	\nonumber
    	&=
    	  \braa{E} \Pi_{\adj} \kett{\rho} \frac{d+1}{d^2-1} \Tr[\Pi_\adj \Pi_Z \Lambda]  \frac{\Tr[(\Pi_\triv + \Pi_Z)\Pi_Z \Lambda ]}{\Tr[\Pi_Z]} \, \lambda_Z(\Lambda)^{m-1}\\\nonumber
    	&=
    	\braa{E} \Pi_{\adj} \kett{\rho} \, \lambda_Z(\Lambda)^{m+1} ,
        \nonumber
\end{align}
where we have used $\Tr[\Pi_\triv \Lambda]=1$, $B=\sum_b \kettbraa{b}{b} = (\Pi_\triv + \Pi_Z)$ and the vanishing projectors identities 
\begin{equation}\label{eq:orthogonality_projectors}
    \Pi_\adj \Pi_\triv = \Pi_Z \Pi_\triv = \Pi_Z \Pi_\ort = \Pi_\triv \Pi_\ort =  0.
\end{equation}

\end{proof}

\begin{proof}[Proof of \cref{lemma:circuit_M_estimator}]
    For shadows build on random a circuit of a single Clifford operator~$c$ followed by $m$ gates $(k_1,\dots,k_m)$ in $\mathbb{K}$, the physical frame operator is
    \begin{align}
        \natM_{m+1} &= \sum_b \EE_{c \in \mc C}\EE_{k_1,\dots,k_m \in \mathbb{K}} \Ad^\dagger(g_\eend) \kettbraa{b}{b} \phi(k_m) \cdots \phi(k_1) \cdot \phi(c) \\
        \nonumber
        &=
        \EE_{c \in \mc C}\EE_{k_1,\dots,k_m \in \mathbb{K}} \Ad^\dagger(c) \Ad(k_\inv) B \Lambda \Ad(k_m) \dots \Lambda \Ad(k_1) \Lambda \Ad(c) \\
        \nonumber
        &= 
        \sum_{\alpha\in \{\triv, Z, \ort \}} \EE_{c \in \mc C} \Ad^\dagger(c) \Pi_\alpha \Lambda \Ad(c) \frac{\Tr[\Pi_\alpha B \Lambda]}{\Tr[\Pi_\alpha]} \frac{\Tr[\Pi_\alpha \Lambda]}{\Tr[\Pi_\alpha]}^{m-1} \\
        \nonumber
        &=
        \sum_{\alpha\in \{\triv, Z, \ort \}} \EE_{c \in \mc C} \frac{\Tr[\Pi_\triv\Pi_\alpha \Lambda]}{\Tr[\Pi_\triv]} \Pi_\triv
        + 
        \sum_{\alpha\in \{\triv, Z, \ort \}} \EE_{c \in \mc C} \frac{\Tr[\Pi_\adj\Pi_\alpha \Lambda]}{\Tr[\Pi_\adj]} \lambda_\alpha(\Lambda)^m \Pi_\adj \\
        \nonumber
        &=
        \Pi_\triv + \frac{\lambda_Z(\Lambda)^{m+1}}{d+1} \Pi_\adj ,
        \nonumber
    \end{align}
    where we have used again $\Tr[\Pi_\triv \Lambda]=1$, $B=\sum_b \kettbraa{b}{b} = (\Pi_\triv + \Pi_Z)$ and the orthogonality relations for projectors in \cref{eq:orthogonality_projectors}.
    Our estimator in \cref{eq:RB_calibrated_frame_onego} corresponds then to $\natM_{m+1}$.
\end{proof}

\begin{protocol}{Noise-mitigated classical shadow estimation}

    \\[3pt]
    
    \textbf{Assumption:} Gate-independent noise channel $\Lambda$ for the Clifford and CNOT-dihedral gate-sets.
    
	\smallskip 

    \textbf{Setup:} Experimentally implemented Clifford gates $\{\phi (g)\}_g$.

    \smallskip 

    \textbf{Input:} Unknown quantum state $\rho$.

	\smallskip

	\textbf{Output:} Noise-mitigated classical shadow $\hat \rho$.
	
	\smallskip
	
	\textbf{Procedure:}
 
    \begin{enumerate}[label=(\roman*)]
        \item Fix a sequence length $m$.
        \item Pick a single random Clifford gate $c$ and $m$ gates from the CNOT-dihedral group, $k_1, \dots, k_m$. 
        \item Implement the sequence $\phi(k_m) \dots \phi(k_1) \phi(c)$ on the target state $\rho$. 
        \item Measure in the computational basis, obtain outcome $b'$. 
        \item Classically compute $(d+1)\braa{E}\Pi_\adj \Ad^\dagger(g_\eend)\kett{b'}$, with $g_\eend= k_m\cdots k_1 \cdot c$.
        \item Repeat the above procedure $N=\mc O(1)$ times to estimate the sample average $\mc Q (m+1)$.
        \item Increase the sequence length $m+1$ and repeat the above procedure.
        \item Fit the data into the model in \cref{shadow_dihedral_fitting_model} to extract $\lambda_Z(\Lambda)$.
        \item Compute the calibrated the frame operator $\hatM_{m+1}$ according to \cref{eq:RB_calibrated_frame_onego}.
        \item Produce the noise-mitigated classical shadows $\hat \rho_{m+1}^{(j)} = \hatM^{-1}_{m+1}\Ad^\dagger(g_\eend) \kett{b'}$ for $j=1,\dots,N$.
    \end{enumerate}

\end{protocol}

\section{Clifford group calibration}

In this section, we turn to discussing in detail the full Clifford group calibration. When the the structure of the noise model is known, at least to the point that we can exclude the contribution of phase errors -- given in the form of $Z$-Pauli strings in $\{\1,Z\}^{\otimes n} \backslash \{\1^{\otimes n}\}$ -- we can employ a standard Clifford RB scheme to calibrate the frame operator. This may be more convenient if we want to use this very common routine and without having to setup a  CNOT-dihedral random circuit.

In the following, we will first describe a calibration scheme based on the global Clifford group as an effective approach to construct noise-mitigated shadows. As we will observe in the following, the crucial point is that the conventional Clifford RB procedure will return $\lambda_\adj(\Lambda) = (\Tr[\Lambda]-1)/({d^2-1})$: as long as the noise channel has negligible $Z$-noise contribution, this quantity will be close to~$\lambda_Z(\Lambda)$, hence providing a reliable correction of the classical shadows. We demonstrate this simpler approach with numerical examples.
Furthermore, we will describe an efficient \emph{local} Clifford calibration scheme for the calibration parameters required for shadow estimation of observables with $\log n$ support size.

\subsection{Global Cliffords}\label{sec:global_clifford}

We recall that the $\Ad$ representation of the Clifford group conveniently splits into two non-degenerate representations only: the trivial and the adjoint. In other words, 
\begin{equation}
\Ad \simeq 1 \oplus \adj. 
\end{equation}
The former carries the one-dimensional subspace spanned by the identity, $\1$, the latter carries the invariant subspace of trace-less matrices, that is,
\begin{equation}
	\mc{S}_\adj = \mathrm{span}\{ p_j \, : \, j=1,\dots,d^2-1\}.
\end{equation}
The fitting model of the RB scheme is thus conveniently expressed by a single parameter decay~\cite{MagGamEmer2,MagGamEmer}
\begin{equation}
	Q_\avg (m,E,\rho)  
	=
	a_\triv + a_\adj \, \lambda_\adj^m
\end{equation}
for some scalars $a_\triv$ and $a_\adj$ absorbing SPAM errors. Here, we have used that
\begin{equation}
\lambda_\triv (\Lambda)= \Tr[\Lambda(\1/d)]=1,
\end{equation}
which follows from the 
trace preserving condition for $\Lambda$,
and
\begin{equation}\label{eq:RB_parameter}
	\lambda_\adj(\Lambda)= \frac{\Tr[\Pi_\adj \Lambda]}{\Tr[\Pi_\adj]}= \frac{\Tr[\Lambda]-1}{d^2-1}.
\end{equation}
According to the formula in \cref{eq:RB_calibrated_frame} we construct the error-mitigated frame operator,
\begin{equation}\label{eq:global_frame_M}
    \hatM_\adj 
    =
    \Pi_\triv + \frac{\lambda_\adj(\Lambda)}{d+1} \Pi_\adj.
\end{equation}

Then, the error of the calibrated shadow procedure can be stated as 
follows, combining \cref{eq_app:natM,eq:global_frame_M}. 

\begin{lemma}[Bias for the calibrated frame operator]\label{thm:diff_calibrated_inv_app}
	It holds that
	\begin{equation} 
	\hatM_\adj^{-1}  \, \natM - \1 
	=
    \left\{\frac{\lambda_Z(\Lambda)}{\lambda_\adj (\Lambda)}  - 1 \right\} \Pi_\adj 
    =
    \left\{(d+1)\frac{\Tr[\Pi_\adj B \Lambda]}{\Tr[\Pi_\adj\Lambda]}  - 1 \right\} \Pi_\adj .	
	\end{equation} 
\end{lemma}
\noindent While this clearly does not correspond to a perfect calibration, there are different scenarios where the gap between the two quantities closes, yielding effective and stable calibration. More precisely, this favorably occurs when the noise channel does not overlap significantly with Pauli strings containing purely Pauli-$Z$ operators and local identities, that is, has negligible support on the invariant subspace connected to the $Z$-parameter. Practical scenarios are:
(i)~random Pauli noise (where the portion of $\{\1,Z\}^{\otimes 2n}$-Pauli strings is $1/2^n$ of the total), (ii)~considering $\Lambda$ to be the depolarizing channel -- where this conventional Clifford RB protocol provides a perfect calibration. This noise model is often sensible, especially for large system sizes (e.g., in the context of achieving a quantum advantage for NISQ devices~\cite{Boixo18,GoogleSupremacy,PhysRevLett.127.180501,SupremacyReview}).

Indeed, previous work~\cite{RQCglobaldep} has shown that for 2-local random quantum circuits whose single-gate noise is sufficiently weak and unital, after $n \log n$ steps the resultant noise is the global depolarizing noise; in other words, random circuits scramble local gate-dependent noise. 
In this context,  if we consider the Clifford group on $n$-qubits to be generated by a local set such as the CNOT, Hadamard and phase gates, any Clifford operator generated by $\mc O (n \log n)$ CNOT gates is an ideal operator equipped with the depolarizing channel. Note that as the system size grows, randomly choosing a Clifford gate prepared by $n \log n$ local CNOT is an event that occurs with probability~1.

\medskip

Like the CNOT-dihedral scheme in \cref{sec:robust_CNOT_dihedral_scheme}, this scheme can also be used in a self-calibrating fashion, that is, producing the shadows and the calibrated frame operator from the same experiments. In this case, the calibrated operator will be
\begin{equation}\label{eq:self-calibrated_Clifford}
    \hatM_{\adj,m} = \Pi_\triv + \frac{\lambda_\adj(\Lambda)^m}{d+1} \Pi_\adj ,
\end{equation}
yielding the same error bias as in \cref{thm:diff_calibrated_inv_app}, and a calibrated shadow obtained from a measurement of a global Clifford circuit of length $m$
\begin{equation}\label{eq:self-calibrated_shadows}
    \hat \rho_{m} = \hatM_{\adj,m}^{-1} \Ad^\dagger (g_\eend) \kett{b'},
\end{equation}
where $g_\eend = c_m \cdots c_1$.

\medskip 

We will observe the validity of this calibration scheme in following numerical investigation.

\begin{protocol}{Global Clifford calibrated frame operator}

    \\[3pt]
    
    \textbf{Assumption:} Gate-independent noise channel $\Lambda$ for the Clifford gate-set with negligible $Z$-noise component.
    
    \smallskip 

    \textbf{Setup:} Experimentally implemented Clifford gates $\{\phi (g)\}_g=\{\Lambda \Ad(g)\}_g$.

    \smallskip 

    \textbf{Input:} Unknown quantum state $\rho$.

	\smallskip

	\textbf{Output:} Calibrated frame operator $\hatM_\adj$.
	
	\smallskip
	
	\textbf{Procedure:}
	 
	 \smallskip
	 
	\begin{enumerate}[label=(\roman*)]
        \item Fix a sequence length $m$.
        \item Choose Clifford elements $c_1,\dots,c_m$ at random.
		  \item Apply a random sequence of Clifford gates with a final inverse gate
		  $\phi(c_m) \phi(c_{m-1}) \dotsb \phi(c_1)$ on the target state $\rho$.
		  \item Measure in the computational basis, obtain outcome $b'$,
        \item Repeat the above procedure $N=\mc O(1)$ times times to estimate the sample average~$\mc Q (m)$.
        \item Increase the value of $m$ and repeat the above.
		  \item Fit the sample averages~$\{\mc Q (m)\}_m$ for different $m$ to the exponential decay model and extract the RB decay parameter~$\lambda_\adj$.
        \item Compute the calibrated frame operator $\hatM_\adj$ according to \cref{eq:global_frame_M} (separated noise-mitigated shadows) or \cref{eq:self-calibrated_Clifford} (self-calibrated shadows as per \cref{eq:self-calibrated_shadows}).
        \end{enumerate}
\end{protocol}

\subsection*{Noise examples}\label{app:noise_examples}

In this section, we will consider three different noise models that are particularly relevant in present-day architectures -- depolarizing, bit-flip and amplitude damping noise -- and derive more specific and detailed 
expressions for the quantities characterizing the bias in \cref{thm:diff_calibrated_inv_app}.

\textbf{Global depolarizing noise.}
This error model can be considered the reference for noisy random quantum circuits, where the noise of the local gates is scrambled into a global depolarizing channel paired with the $n$-qubit gate resulting from the implementation~\cite{RQCglobaldep}.

\begin{lemma}[Calibration for global depolarizing channel, cnf.~\cite{koh2022ClassicalShadowsNoise}]
	Considering a $n$-qubit global depolarizing channel with parameter $p$, namely, $\phi(g) = \Lambda_\mathrm{dep} \Ad (g)$ with 

    \begin{equation}
        \Lambda_\mathrm{dep}(\rho) = (1-p)\rho + p \frac{\1}{d} \, ,
    \end{equation}
	we then have
	\begin{equation}
        \hatM_\adj^{-1}\, \natM - \1 = 0 .
    \end{equation}
    
\end{lemma} 

Thus, the calibrated estimation is exact, whereas the uncalibrated procedure will have an error~$\mc O (p)$.  

\begin{proof}
	This follows immediately from explicitly evaluating $\Tr(\Lambda_\mathrm{dep}) = 1 + (d^2-1)(1-p)$ and 
 	\begin{equation}
 \lambda_Z(\Lambda_\mathrm{dep}) = 1-p, 
     \end{equation}
     so that
    $\lambda_\mathrm{adj}(\Lambda_{\mathrm{dep}}) = \lambda_Z(\Lambda_{\mathrm{dep}} ) = 1 - p$.

\end{proof}

\newpage

\textbf{Local depolarizing noise.} In this setting we have the following.

\begin{figure}
    \centering
\includegraphics[width=0.55\textwidth]{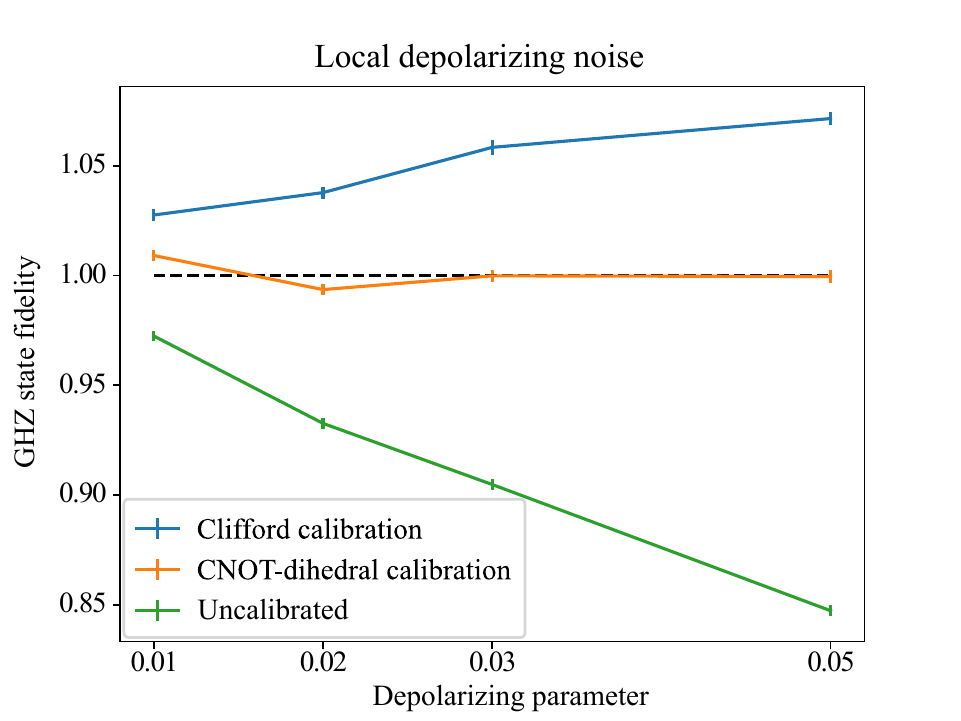}
    \caption{Simulation results for 5 qubits and local depolarizing noise. In line with the results of \cref{eq:local_depol_values}, the calibration using the Clifford RB parameter overcompensates the effect of the noise, since $\lambda_\mathrm{adj}(\Lambda_\mathrm{dep,loc}) < \lambda_Z(\Lambda_\mathrm{dep,loc})$. See \cref{sec:numerics_details} for details on the numerical simulation.}
    \label{fig:local_depol}
\end{figure}

\begin{lemma}[Local depolarizing noise]
	We consider the noise channel $\Lambda_\mathrm{dep,loc}$, in which each qubit is individually affected by a single-qubit depolarizing channel with parameter $p$. Then
	\begin{equation}
		\lambda_Z(\Lambda_\mathrm{dep,loc}) = \frac{(2-p)^n - 1}{d-1}  \qquad \text{and} \qquad  \lambda_\mathrm{adj}(\Lambda_\mathrm{dep,loc}) = \frac{(4-3p)^n - 1}{d^2 - 1},
  \label{eq:local_depol_values}
	\end{equation}
	and thus 
	\begin{equation}
		\hatM_\adj^{-1} \, \natM - \1 
		=
		\left\{(d+1)\frac{(2-p)^n -1}{(4-3p)^n -1}  - 1 \right\} \Pi_\adj.
	\end{equation} 
\end{lemma}

\begin{proof}
	We expand the expression for the $n$-qubit local depolarizing channel as
	\begin{align}
		\Lambda_\mathrm{dep,loc}
		&=
    	(1-p)^n \1_{2n} + (1-p)^{n-1}p \frac{1}{4} \sum_{\alpha_1 \in      \{0,1,2,3\} } \sum_{j=1}^n g_{\alpha_1}^{(j)} \otimes \1^{j^c}_{n-1} \otimes g_{\alpha_1}^{(j)} \otimes \1^{(j^c)}\nonumber\\
    		&+
    		(1-p)^{n-2}p^2 \frac{1}{4^2} \sum_{\alpha_2 \in \{0,1,2,3\}^{\times 2} } \sum_{j \neq k}  g_{\alpha_2}^{(j,k)}  \otimes \1^{(j,k)^c} \otimes g_{\alpha_2}^{(j,k)}\otimes  \1^{(j,k)^c} \\
    		&+
    		\dots+ p^n \frac{1}{4^n} \sum_{\alpha_n \in \{0,1,2,3\}^{\times n} } g_{\alpha_n}\nonumber \\
    		&=
    		\sum_{\ell=0}^n (1-p)^{n-\ell} p^{\ell} \frac{1}{4^\ell} \sum_{\alpha_\ell \in \{0,1,2,3\}^{\times \ell} } \sum_{\mu \in  \mc{W} (n,\ell) }  g_{\alpha_\ell}^{\mu} \otimes \1^{\mu^c} \otimes g_{\alpha_\ell}^{\mu} \otimes \1^{\mu^c}\nonumber,
        \end{align}
	where  we denote $\mc W (n,\ell)$ to be the set of combinations of  $\ell$ qubits out of $n$, $X^w$, $g_{\alpha_\ell}^{\mu}$ is a Pauli string fo length $\ell$ (indexed by $\alpha_\ell$) acting on subsystem $\mu$ and where the identity is acting on the complement system $\mu^c$.
	Since every Pauli string made of identities and $Z$-operators only will give a contribution of $2^n$, we then have
	\begin{eqnarray}
		\lambda_Z(\Lambda_\mathrm{dep,loc})
		&=&
		\frac{1}{d-1} \left(\sum_{b \in \{0,1\} }^n \braa{b} \Lambda_\mathrm{dep,loc} \kett{b}  - 1 \right)
		 \\ \nonumber 
		&=&
		\frac{1}{d-1} \left( \sum_{\ell =0}^n \left[\binom{n}{\ell} (1-p)^{n-\ell} \left(\frac{p}{2}\right)^\ell \right] - 1 \right)
		=
		\frac{(2-p)^n - 1}{d-1}.
	\end{eqnarray}
	Conversely, only the 2$n$-qubit identity string will contribute to the trace, hence following an analogous calculation yields
	\begin{equation}
		\Tr[\Lambda_\mathrm{dep,loc}] =(4-3p)^n .
	\end{equation}

\end{proof}

\textbf{Bit-flip noise.} 
For this noise type, the Clifford RB calibration is very effective. Indeed, we find the following.

\begin{figure}
	\centering
 \includegraphics[width=0.55\textwidth]{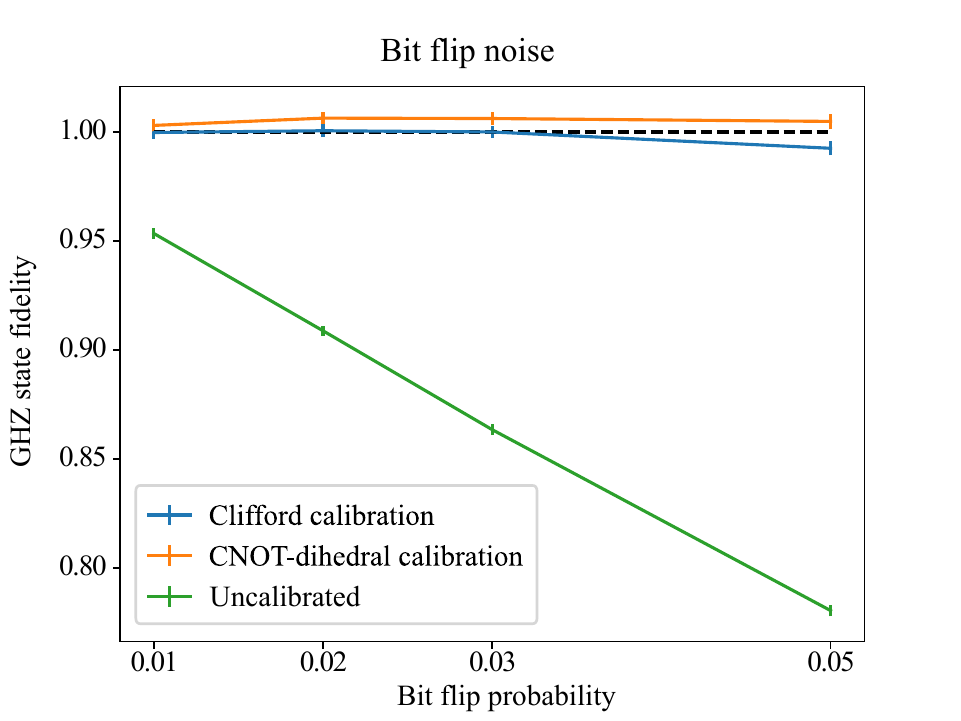}
	\caption{Simulation results for 5 qubits and local bit flip noise. See \cref{sec:numerics_details} for details on the numerical simulation.}
	\label{bitflip_plot}
\end{figure}

\begin{lemma}[Bit-flip noise]
Consider the bit-flip channel $\Lambda_\mathrm{flip}$, where each qubit is errorless with probability $1-p$ and experiences Pauli-$X$ noise with probability $p$. Then
\begin{equation}
	\lambda_Z(\Lambda_\mathrm{flip}) = \frac{d(1-p)^n - 1}{d-1} \qquad \text{and} \qquad  \lambda_\mathrm{adj}(\Lambda_\mathrm{flip}) = \frac{d^2 (1-p)^n - 1}{d^2-1} ,
\end{equation}
and therefore,
\begin{equation}
	\hatM_\adj^{-1} \,  \natM - \1 
	=
	\left\{ (d +1)\frac{d \, (1-p)^n -1}{d^2(1-p)^n - 1}  -1 \right\} \Pi_\adj .
\end{equation} 
\end{lemma}

\begin{proof}
	Following the above logic, we expand the expression for the bit-flip noise channel as
\begin{aligneq}
	\Lambda_\mathrm{flip}
	&=
	(1-p)^n \1_{2n} + (1-p)^{n-1}p \sum_{j=1}^n X^{(j)} \otimes \1^{j^c}_{n-1} \otimes X^{(j)} \otimes \1^{(j^c)}\\
	&+
	(1-p)^{n-2}p^2 \sum_{j \neq k} X^{(j,k)}  \otimes \1^{(j,k)^c} \otimes X^{(j,k)}\otimes  \1^{(j,k)^c}\\
	&+
	\dots+ p^n X^{\otimes 2n}  \\
	&=
	\sum_{\ell=0}^n (1-p)^{n-\ell} p^{\ell} \sum_{\mu \in  \mc{W} (n,\ell) } X^{\mu} \otimes \1^{\mu^c} \otimes X^{\mu} \otimes \1^{\mu^c},
\end{aligneq}
where $X^w$ is the application of $X$ gates on the qubits in $\mu$ and the identity acting on the complement system $\mu^c$.
We have in this case
\begin{equation}
	\lambda_Z(\Lambda_\mathrm{flip})
	=
	\frac{1}{d-1} \left( \sum_{b \in \{0,1\} }^n \braa{b} \Lambda_\mathrm{flip} \kett{b} - 1 \right)
	= \frac{d (1-p)^n - 1}{d-1}  ,
\end{equation}
since any string containing a Pauli $X$ operator will vanish. Analogously, only the 2$n$-qubit identity element will contribute to the trace, hence
\begin{equation}
	\Tr[\Lambda_\mathrm{flip}] = d^2 (1-p)^n .
\end{equation}
\end{proof}

\textbf{Amplitude damping noise.}
This example can be considered a more complicated testbed for the calibration procedure, since the noise channel has an overlap with Pauli-$Z$ operators.

\begin{figure}
	\centering
	\includegraphics[width=0.55\textwidth]{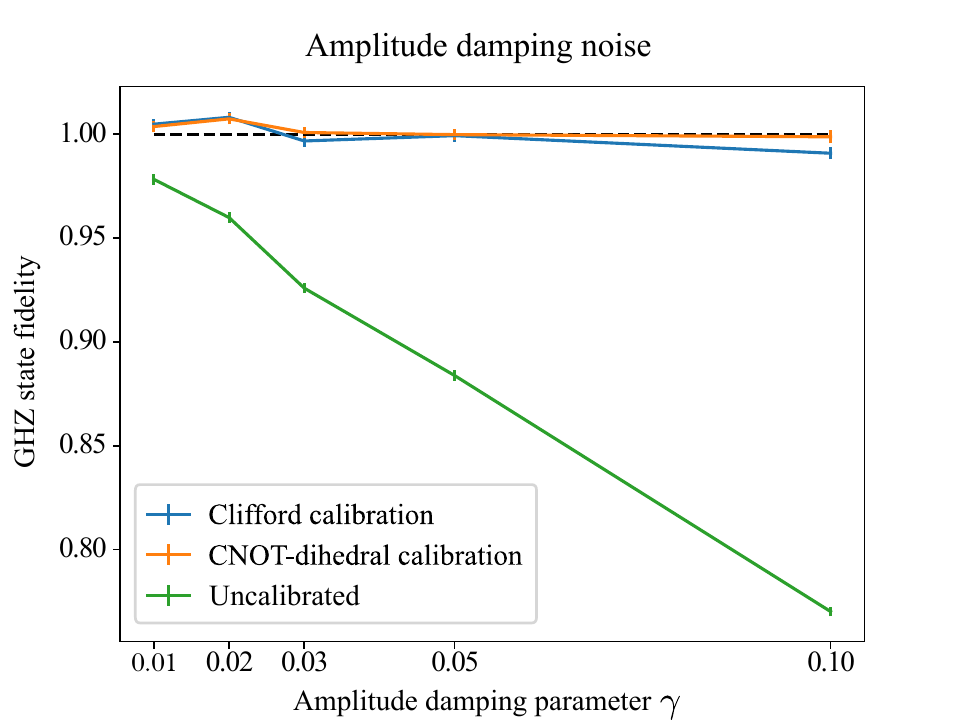}
	\caption{The performance for a 5 qubit system and amplitude damping noise. See \cref{sec:numerics_details} for details on the numerical simulation.}
	\label{ampdamp_plot}
\end{figure}

\begin{lemma}[Amplitude damping noise]\label{ex:amp_damp_noise}
Consider the amplitude damping channel $\Lambda_\mathrm{AD}$, where each qubit is errorless with probability $p$, and damped with probability $1-p$ according to a damping parameter~$\gamma$. Then
\begin{equation}
	\lambda_Z(\Lambda_\mathrm{AD} )= \frac{ d + (2-p+(p-1)\gamma)^n - (2-p)^n - 1}{d-1}  \qquad 
\end{equation}
and
\begin{equation} 
	\lambda_\mathrm{adj}(\Lambda_\mathrm{AD}) = \frac{\big[  2-\gamma +2 \sqrt{1-\gamma} + p (2 + \gamma - 2 \sqrt{1-\gamma}) \big]^n - 1}{d^2 - 1} .
\end{equation}
\end{lemma}
\begin{proof} 
As before, we expand the expression for the amplitude damping channel as
\begin{align} 
	\Lambda_\mathrm{AD}
	&=
	\sum_{\ell=0}^n p^{n-\ell} (1-p)^{\ell} \sum_{\mu \in  \mc{W} (n,\ell) } 
	\Big[E_1^{\mu} \otimes \1^{\mu^c}_{n-\ell} \otimes E_1^{\mu} \otimes \1^{\mu^c}_{n-\ell}
	+
	E_2^{\mu} \otimes \1^{\mu^c}_{n-\ell} \otimes E_2^{\mu} \otimes \1^{\mu^c}_{n-\ell} \Big] ,
\end{align}
where $E_1$ and $E_2$ are the amplitude damping operators.
We then have
\begin{align}
	\lambda_Z(\Lambda_\mathrm{AD})
	&= \frac{
	\sum_{b \in \{0,1\} }^n \braa{b} \Lambda_\mathrm{AD} \kett{b} - 1}{d-1} \\ 
	\nonumber 
	&= \frac{
	\sum_{b \in \{0,1\} }^n \bra{b} \Lambda_\mathrm{AD}^{\text{first copy}} \ket{b}^2 - 1}{d-1} \\ \nonumber 
	&= \frac{1}{d-1} \left(
	\sum_{k=0}^n \binom{n}{k} p^{n-k} (1-p)^k \Big(2^n + \sum_{\ell = 0}^k \binom{k}{\ell} ((\sqrt{1-\gamma})^{2\ell}-1) \Big) - 1 \right) \\  \nonumber 
	&= \frac{1}{d-1} \left(
	d - 1 +  \sum_{k=0}^n \binom{n}{k} p^{n-k} (1-p)^k  (2-\gamma)^k - \sum_{k=0}^n \binom{n}{k} p^{n-k} (1-p)^k 2^k \right) \\ \nonumber 
	&= \frac{d +  [p+(1-p)(2-\gamma)]^n - [p+2(1-p)]^n - 1}{d-1}
	 \\ \nonumber 
	&=
	\frac{d + [2-p+(p-1)\gamma]^n - [2-p]^n - 1}{d-1}
\end{align}
and
\begin{align}
	\Tr[\Lambda_\mathrm{AD}]
	&=
	\sum_{k=0}^n \binom{n}{k} p^{n-k} (1-p)^k 4^{n-k} (2-\gamma +2 \sqrt{1-\gamma})^k \\ 
	\nonumber 
	&=
	\big[4p+ (2-\gamma +2 \sqrt{1-\gamma})(1-p)\big]^n \\ 
	\nonumber 
	&=
	\big[  2-\gamma +2 \sqrt{1-\gamma} + p (2 + \gamma - 2 \sqrt{1-\gamma}) \big]^n , \nonumber 
\end{align}
from which the claim follows.
\end{proof}
The asymptotic behavior for many-qubit systems is not evident from \cref{ex:amp_damp_noise}, but we can observe from \cref{ampdamp_plot} that for an 8-qubit system, the calibration yields to an error-robust fidelity estimation even for significant levels of noise damping (represented by the coefficient~$\gamma$).

\subsection{Local Clifford scheme}\label{sec:local_Clifford}

A scalable, noise-mitigated protocol can also be obtained for shadows of a quantum state generated by the local Clifford group. Local here refers to the group of single qubit Cliffords. From a practical perspective, this group is by far distinguished: In practice, this can be routinely
implemented for a number of architectures. Here, we write the real frame operator in terms of the tensor of local projections $\Pi_w$ labeled by strings $w \in \{0,1\}^{\times n}$ containing $\abs{w}$-many~'1's. That is, 
we find the following expression for the physical frame operator (which was already derived in Ref.~\cite{RobustShadows}).

\begin{lemma}[Frame operator for local Clifford shadow estimation]
\begin{equation}\label{eq:natM_local_Clifford}
	\natM = \sum_w \frac{\Tr[\Pi_w B \Lambda]}{\Tr[\Pi_w]} \Pi_w \eqqcolon \sum_w c_w \Pi_w,
\end{equation}
with
\begin{equation}\label{coeff_c_w}
	c_w = \frac{1}{2^n3^{\abs{w}}} \sum_{x,y\in \{0,1\}^n } (-1)^{w \cdot (x \oplus y)} \bra{y,y}\Lambda \ket{x,x}.
\end{equation}
\end{lemma}

\begin{proof}
To write $\Pi_w$ in a convenient form we do some tensor 
manipulation and transform 
\begin{equation}\label{eq:tensor_transformation}
	\kett{b} = \ket{b_1,b_2,\dots,b_n; b_1,b_2,\dots,b_n} \mapsto \ket{b_1,b_1,b_2,b_2,\dots,b_n,b_n}.
\end{equation}
Under this isomorphism, we can consider for each set of projections with the same Hamming weight $\abs{w}$ a representative element 
$\Pi_{\abs{w}}^\mathrm{rep} = P_\loc^{\otimes n - w} \otimes (\1_4 -P_\loc)^{\otimes w}$ with $P_\loc = \sum_{j,k \in \{0,1\}} \ketbra{j,j}{k,k}/2$. 
We now compute 
\begin{align}
	&\Pi_{\abs{w}}^\mathrm{rep} B 
	=
	P_\loc^{\otimes n - w} \otimes (\1_4 -P_\loc)^{\otimes w} \sum \ketbra{b_1,b_1,\dots,b_n,b_n}{b_1,b_1,\dots,b_n,b_n} \\
	\nonumber
	&=
	\sum_{j,k,b \in \{0,1\}} \Big( \frac 1 2 \kett{j} \braakett{k}{b}\braa{b}\Big)^{\otimes n - w}
	\otimes
	\Big( \sum_{b \in \{0,1\}} \kettbraa{b}{b} -\frac 1 2 \sum_{m,\ell,b} \ket{\ell,\ell}\braakett{m}{b}\braa{b} \Big)^{\otimes w} \\
	\nonumber
	&=
	\frac{1}{2^{n-w}} \Big( \sum_{j,b \in \{0,1\}} \kettbraa{j}{b}\Big)^{n-w}
	\otimes
	\Big( \sum_{b \in \{0,1\}} \kettbraa{b}{b} -\frac 1 2  \sum_{\ell,b \in \{0,1\}} \kettbraa{\ell}{b} \Big)^{\otimes w} \\
	\nonumber
	&=
	\frac{1}{2^{n-w}} \Big( \sum_{j,b \in \{0,1\}} \kettbraa{j}{b}\Big)^{n-w}
	\otimes
	\frac {1}{2^w} \Big( \sum_{b \in \{0,1\}} \kettbraa{b}{b} -\sum_{\ell,b \in \{0,1\}} \kettbraa{b\oplus 1}{b} \Big)^{\otimes w} \\
	\nonumber
	&=
	\frac{1}{2^n} \sum_{x,y \in \{0,1\}^n } \kettbraa{x}{y} (-1)^{(\{0\}^{\times n-w}\times \{1\}^{\times w} )\cdot (x \oplus y)},
	\nonumber
\end{align}
where $x \oplus y$ is a component-wise addition modulo~2.
From this calculation, we can then see that for a general projection $\Pi_w$ for some label $w$ we have 
\begin{equation}
	\Pi_w B = \frac{1}{2^n} \sum_{x,y \in \{0,1\}^n } \kettbraa{x}{y} (-1)^{w \cdot (x \oplus y)},
\end{equation}
and consequently (with $\Tr[\Pi_w]=3^{\abs{w}}$)
\begin{equation}
	c_w =\Tr[\Pi_w B \Lambda]/3^{\abs{w}} = \frac{1}{2^n3^{\abs{w}}} \sum_{x,y\in \{0,1\}^n } (-1)^{w \cdot (x \oplus y)} \braa{y}\Lambda \kett{x}.
\end{equation}
\end{proof}

We shall now consider the quantity that we can retrieve from the local Clifford gate-set shadow procedure. From Eq.~(15) of Ref.~\cite{RandomSequences}, this is,
\begin{equation}\label{eq:loc_fidelities}
	p_{w,B}(\Lambda) = \Tr[\Pi_w B \Pi_w \Lambda]/3^{\abs{w}}
    =
    \Tr[B_w \Lambda]/3^{\abs{w}},
\end{equation}
where $B_w = \Pi_w B \Pi_w$.
We note that, for $B_\loc \coloneqq \kettbraa{0}{0} + \ketbra{1,1}{1,1}$,
\begin{equation}\label{eq:loc_proj_B}
	B_\loc P_\loc = B_\loc P_\loc = P_\loc
	\quad
	\text{and}
	\quad
	(\1_4-P_\loc)B_\loc(\1_4-P_\loc) = (\1_4-P_\loc)B_\loc.
\end{equation}
In our basis representation \cref{eq:tensor_transformation} $B=B_\loc^{\otimes n}$ and thus from \cref{eq:loc_proj_B} follows that $\Pi_w B \Pi_w = \Pi_w B$ and consequently
\begin{equation}
    p_{w,B}(\Lambda) = c_w
\end{equation} 
as defined in \cref{coeff_c_w}. 

The sample complexity is controlled by Theorem~4, Eq.~(16) in Ref.~\cite{RandomSequences}. 
Since 
\begin{equation}
    \Tr[B_w B_w^\dagger] = \Tr[B_w]= 1,
\end{equation}
 we have $p_{w,B_w}(\Lambda) = \mc O (3^{-\abs{w}})$. In the regime $\abs{w} =\mc O (\log n)$, we will require an estimation of this quantity up to precision $1/\mathrm{poly}(n)$. According to the variance bounds in Eq.~(16) of Ref.~\cite{Shadows} we have scalable protocol in~$n$. Note that here we did not assume any local structure for $\Lambda$.

From the above considerations, it follows that again we recover precisely  (exponentially) many coefficients $c_w$ with $\abs{w}=\mc O (\log n)$  characterizing $\natM$.
If we want to learn a set of functions $\Tr[O_i \rho]$ with individual support that can be broken down into parts~$\mc O(\log n)$, it suffices to recover only those coefficients whose support is contained in it. 
Hence, we can efficiently build a calibrated frame operator 
\begin{equation}\label{eq:loc_frame_operator}
    \hatM_\loc (\{O_i\}_i) \coloneqq \sum_{w\  :\ \supp{w} \subseteq \ \bigcup\supp{O_i}} p_{w,B}(\Lambda) \, \Pi_w,
\end{equation}
 which yields a error-robust local Clifford shadow scheme (under the gate-independent noise assumption).

\begin{protocol}{Local Clifford calibrated frame operator}

    \\[3pt]
    
    \textbf{Assumption:} Gate-independent noise channel $\Lambda$ for the local Clifford gate-set.
    
    \smallskip 

    \textbf{Setup:} Experimentally implemented local Clifford gates $\{\phi (h)\}=\{\Lambda\Ad(h)\}$.

    \smallskip 

    \textbf{Input:} Unknown quantum state $\rho$, set of observables $\{O_i\}$ with $\supp O_i = \mc O(\log n)$.

	\smallskip

	\textbf{Output:} Calibrated frame operator  $\hatM_\loc$.
	
	\smallskip
	
	\textbf{Procedure:}
	  \smallskip
	 
	\begin{enumerate}[label=(\roman*)]
        \item Fix a sequence length $m$.
        \item Choose Clifford elements $h_1,\dots,h_m$ at random.
		  \item Apply a random sequence of Clifford gates
		  $\phi(h_m) \dotsb \phi(h_1)$ on the target state $\rho$.
		  \item Measure in the computational basis, obtain outcome $b'$.
        \item For every $w$ such that $w \subseteq \supp O_i$ for some~$i$, classically compute 
        $\braa{b'}\adj(h_m) \prod_{j=1}^{m-1} B_w \adj(h_j)\kett{\vartheta}$, where $\adj$ is the irreducible adjoint representation of the $n$-copy 1-qubit Clifford group and $\vartheta$ is some probe state (e.g., an ideal preparation of $\rho$).
        \item Repeat the above procedure $N=\mc O(1)$ times to estimate the sample average~$\mc Q (m)$.
        \item Increase the value of $m$ and repeat the above.
		  \item Fit the sample averages for different $m$ to the exponential decay model from Ref.~\cite{RandomSequences} and extract the RB decay parameter $p_{w,B}=\Tr[B_w \Lambda]/3^{\abs{w}}$.	
        \item Compute the calibrated frame operator $\hatM_\loc (\{O_i\}_i)$ according to \cref{eq:loc_frame_operator}.
        \end{enumerate}
\end{protocol}
\section{Variance bounds}\label{app:variance}
The variance of the calibrated shadow estimation procedure is asymptotically independent from the system dimension for observables having local support. The calculations are carried out under the gate-independent noise assumption.
We recall that the variance is the determining figure for characterizing the sample complexity of both shadow estimation~\cite{Shadows} and RB~\cite{helsen2019multiqubit}.

\subsection{CNOT-dihedral shadow scheme}\label{app:variance_dihedral}
We have proposed a procedure to construct error-mitigated shadows through CNOT-dihedral benchmarking.
This has lead to the  average outcome 
\begin{equation}
\mc Q_\avg(m) = \sum_x \mathbb{E}_{C\in \mc C} \mathbb{E}_{k_1, \ldots k_m\in \KK}  \braa{O}\mc{M}^{-1}\Ad(g_\eend)^\dagger \kett{x}\braa{x} \Lambda \Ad(k_m) \cdots \cdots \Lambda \Ad(k_1) \Lambda \Ad(c)\kett{\rho}
\end{equation}
at sequence length $m$. As we have derived in \cref{shadow_dihedral_fitting_model}, this average is given by
\begin{equation}
\mc Q_\avg(m) = \braa{E} \Pi_{\adj} \kett{\rho} \, \lambda_Z(\Lambda)^{m+1} .
\end{equation}
Hence one can estimate $\braakett{O}{\rho}$ in a robust way by measuring $\mc Q_\avg(m)$ for a few points and calculating the intercept with zero. This can be done efficiently~\cite{harper2019statistical} provided the variance of estimating $p(m)$ is bounded for all $m$. We will see that this straightforwardly is the case. For convenience we assume that~$O$ is traceless.  
We then have the following theorem.

\begin{theorem} Let $O$ be traceless. Then the variance of the CNOT-dihedral shadow scheme (at sequence length $m$) is bounded by 
$\mathbb{V} \leq c \, \Tr(O^2)$
for some constant $c$ independent of~$n$ and~$m$.
\end{theorem}

\begin{proof} 
We find
\begin{align}
\mathbb{V} &\leq (2^n+1)^2\sum_x \mathbb{E}_{c_\eend\in \mc C} \mathbb{E}_{k_1, \ldots k_m\in \KK} \braa{O}\Ad^\dagger(c_\eend) \kett{x}^2 \braa{x} \Lambda \Ad(k_m) \cdots \Lambda \Ad(k_1) \Lambda \Ad(c)\kett{\rho}\\
&= (2^n+1)^2\sum_x \sum_C \braa{x\tn{3}} (I\tn{2} \otimes \Lambda T_{\mc K}(\Lambda)^m ) C\tn{3} \kett{O\tn{2}\otimes \rho}\\
\nonumber
 &= (2^n+1)^2\sum_{\pi,\pi' \in S_3} W_{\pi, \pi'} \braa{x\tn{3}} (I\tn{2} \otimes \Lambda T_{\mc K}(\Lambda)^m )\kett{\pi} \braa{\pi'}\kett{O\tn{2}\otimes \rho}\\
 \nonumber
 &\leq (2^n+1)^2 \tr(O^2) \norm{W}_{1, \infty}\sum_x \big(3\bra{x}\!\Lambda T_{\mc K}(\Lambda)^m(\ket{x}\!\!\bra{x})\!\ket{x} + 3\bra{x}\!\Lambda T_{\mc K}(\Lambda)^m(I)\!\ket{x} \big) 
 \nonumber
 \end{align}
by H{\"o}lder's $1-\infty$ inequality and the fact that $\Lambda T_{\mc K}(\Lambda)^m$ is completely positive. Since $\norm{W}_{1, \infty} = \mc O(2^{-3n})$ and $\Lambda T_{\mc K}(\Lambda)^m$ is also trace preserving, it is straightforward to see that $\mathbb{V} = \mc O(1)$ in both~$n$ and~$m$.
By standard linear regression we can learn the intercept in with variance bounded by $\mc O(\mathbb{V}) = \mc O(1)$ in~$n$. 
\end{proof}

\subsection{Global Clifford shadow scheme}

For completeness, we write down the variance for the noise-mitigated shadow scheme assisted by conventional global Clifford RB. Note that this is a straightforward adaptation of the results of Refs.~\cite{Shadows,RobustShadows}. 

\begin{theorem}[Variance bound]
	The variance of a single-round estimation $\hat \xi \equiv \hat \xi(O,\rho,g,b)= \braakett{O}{\hat \rho}$ is upper bounded by
	\begin{align}
		\mathrm{Var[\hat \xi]} 
		&\leq 
		\frac{(2^n+1)2^{2n-1}}{\lambda_\RB^2(\adj) (2^n-1)(2^{2n}-4)} \Big(\Tr^2[O]+\Tr[O^2]\Big) \\
		\nonumber
		&=
		\mc O (1) \Big(\Tr^2[O]+\Tr[O^2]\Big),  
	\end{align}
	implying that for any local observable $O$ the variance is dimensionless.
\end{theorem}
\begin{proof}
We find
	\begin{align}
		\mathrm{Var[\hat \xi]}
		&\leq 
		\EE_\GG \sum_b \braa{O}^{\otimes 2} (\hatM^{-1})^{\otimes 2} (\Ad^\dagger(g))^{\otimes 2} \kett{b}^{\otimes 2} \braa{b}\phi(g)\kett{\rho} \nonumber\\
		&\leq
		\norm{\hatM^{-1} }_2^2\EE_\GG \sum_b \braa{O}^{\otimes 2} (\Ad^\dagger(g))^{\otimes 2} \kett{b}^{\otimes 2} \braa{b}\Lambda \Ad(g)
		\kett{\rho}.\label{eq:variance_1}
	\end{align}
Now we consider $\Ad^\dagger(g))^{\otimes 2} \kett{b}^{\otimes 2} \braa{b}\Lambda\Ad(g)$ in the  (column-stacking) vectorized form and turn to Weingarten calculus to compute the unitary 3-design,
	\begin{align}
		\EE_\GG \Ad(g)^{\otimes 3} \big(\kett{b}^{\otimes 2} \otimes \Lambda\kett{b}\big) 
		&= 
		\sum_{\pi,\pi'\in \mc S_3} w_{\pi,\pi'} \kett{\pi}\!\braa{\pi'} (\kett{b}^{\otimes 2} \otimes \Lambda\kett{b}) 
		 \label{eq:3-design_WG}
	\end{align}
The coefficients $w_{\pi,\pi'}$ are well characterized (see e.g.~\cite{collins2022weingarten}). Notably, the diagonal elements of $w$ dominate the off diagonal elements, so in the high $n$ limit it suffices to only consider these.  Thus,
	\begin{equation}
		\cref{eq:3-design_WG} = \frac{2^{2n}-2}{2^n(2^{2n}-1)(2^{2n}-4)}\sum_{\pi\in \mc S_3} \Tr[\pi (\kett{b}^{\otimes 2} \otimes \Lambda\kett{b})] \kett{\pi} + O(2^{-n})
	.\end{equation}
Now we can go back to \cref{eq:variance_1} and recall that 
	\begin{align}
		\hatM^{-1} = \frac{2^n+1}{\lambda_\adj} 
	\end{align}
and building on Lemma 7 of Ref.~\cite{PartialDerandomization}, we write
	\begin{align}
		\mathrm{Var[\hat \xi]}
		&\leq 
		\frac{(2^n+1)^2}{\lambda_\RB(\adj)^2} \frac{2^{2n}-2}{2^n(2^{2n}-1)(2^{2n}-4)} 
		\sum_b\sum_{\pi \in \mc S_3} \tr[\pi(O^{\otimes 2}\otimes \rho)] \Tr[\pi (\kett{b}^{\otimes 2} \otimes \Lambda\kett{b})] \\
		\nonumber
		&\leq
		\frac{(2^{n}+1)2^{2n}}{6 \lambda_\RB(\adj)^2 2^n(2^{n}-1)(2^{2n}-4)} \sum_b \Big(
		\Tr[O]^2\Tr[\rho]\Tr[\ketbra{b}{b}]^2\Tr[\Lambda(\ketbra{b}{b})]\\
		\nonumber
		&+
		\Tr[O^2]\Tr[\rho]\Tr[\ket{b}\!\!\braket{b}{b}\!\!\bra{b}]\Tr[\Lambda(\ketbra{b}{b})]
		+
		2 \Tr[O]\Tr[O \rho]\Tr[\ketbra{b}{b}]\Tr[\ketbra{b}{b} \Lambda(\ketbra{b}{b})] \\
		\nonumber
		&+
		2 \Tr[O^2\rho] \Tr[\ket{b}\!\!\braket{b}{b}\!\!\bra{b}\Lambda(\ketbra{b})]
		\Big) \\
		\nonumber
		&\leq
		\frac{(2^{n}+1)2^{2n}}{6\lambda_1^2(\adj)2^n(2^{n}-1)(2^{2n}-4)} \sum_b \Big(
		\Tr[O]^2\Tr[\Lambda(\ketbra{b}{b})]
		+
		\Tr[O^2]\Tr[\Lambda(\ketbra{b}{b})] \\
		\nonumber
		&+
		2 \Tr[O]\Tr[O\rho] \bra{b}\Lambda(\ketbra{b}{b})\ket{b} 
		+
		2 \Tr[O^2\rho] \bra{b}\Lambda(\ketbra{b}{b})\ket{b}
		\Big) \\
		\nonumber
		&\leq
		\frac{(2^{n}+1)2^{2n}}{6\lambda_1^2(\adj)2^n(2^{n}-1)(2^{2n}-4)}   \Big(\Tr[O^2]+\Tr[O]^2 \Big)\Big(\Tr[\Lambda(\1)] + 2 \sum_b \bra{b}\Lambda(\ketbra{b}{b}{b})\ket{b} \Big),
		\nonumber
	\end{align}
where in the last equality we have used $\Tr[O\rho]\leq \Tr[O]$ and  $\Tr[O^2\rho]\leq \Tr[O^2]$. 
Now, by the fact that 
$\sum_b \bra{b}\mc C(\ketbra{b}{b})\ket{b} \leq 2^n$,
we conclude that
\begin{align}
	\mathrm{Var[\hat \xi]}
	&\leq
	\frac{(2^{n}+1)2^{2n}}{6\lambda_1^2(\adj)2^n(2^{n}-1)(2^{2n}-4)}   (2^n+2^{n+1})\Big(\Tr[O^2]+\Tr[O]^2 \Big) \\
	&=
	\frac{(2^{n}+1)2^{2n-1}}{\lambda_1^2(\adj)(2^{n}-1)(2^{2n}-4)} \Big(\Tr[O^2]+\Tr[O]^2 \Big).\nonumber
\end{align}
\end{proof}

\section{Gate-dependent noise}
\label{app:gate_dependence}

In \cref{app:frame} we have observed that under the gate-independent noise assumption the noise is completely confined in the dominant subspaces, hence allowing us to fully determine $\natM$~as per \cref{eq_app:natM} using the schemes presented in \cref{sec:CNOT-dihedral}.
In the more general case of gate-dependent noise however, our estimator $\hatM$ as per \cref{eq:RB_calibrated_frame} differs from the actual physical frame operator~$\natM$ also regarding the invariant subspaces  (which are also perturbed by noise) in the Fourier transform picture presented in the main letter. 
Unfortunately, RB is not able to capture information structure of the invariant eigenspace themselves. 
We know, however, that while basis vectors of a degenerate subspace are ill-conditioned against perturbations (in the sense that the eigenvectors produced by a perturbation may completely differ from the chosen basis), the invariant  subspaces themselves are, however, robust. A perturbed subspace can be obtained from a transformation of the original one according to some bounded operator~\cite{StewartSun,KatoPTLO,frameworkRB}. 
Secondly, we have not taken into consideration the subspace of $\phi_\FT (\alpha)$ corresponding to the perturbation of the kernel of $\Ad_{\FT}[\alpha]$. We simply set all of its eigenvalues to~0 in our estimator. 
Again, we know that each eigenvalue will be within the ball of the perturbation, but this bound alone is in general not sufficient to control the error in estimating $\natM$ for general noise. 

Recent work~\cite{brieger2023StabilityClassicalShadows} has shown that, in some settings, 
shadow estimation might be unstable in the presence of gate-dependent noise and 
robust shadow estimation procedures can even increase a noise-induced bias in the estimation compared to uncalibrated shadows estimation, rather then eliminating it. 
For large classes of observables already the uncalibrated shadow estimation protocol is shown to be (perhaps surprisingly) robust in the presence of gate-dependent noise, in the sense that the bias in the estimation does not scale with the system size.
To investigate how our proposed protocol performs with gate-dependent noise against the uncalibrated scheme, we numerically simulated several cases of simple gate-dependent noise models. 
In all simulations, we observed that our calibration reduces the noise-induced bias.

\medskip 

In \cref{fig:gate_dependence}, numerics for two further gate-dependent noise models are shown, in addition to panel~(b) of Fig. 2 in the main letter. Panel~(a) again shows the performance under gate-dependent depolarizing noise: In addition to 2-qubit depolarizing noise following CNOT gates, each single-qubit gate experiences local depolarizing noise with a parameter that is $10\%$ of the 2-qubit noise. Although the gate-dependence diminishes the overall performance, the CNOT-dihedral calibration still significantly reduces the noise-induced bias. As expected from the discussion on local depolarizing noise in \cref{sec:global_clifford}, the Clifford calibration overcompensates the noise, increasingly so for higher levels of noise.

Panel (b) assumes that the shadow estimation unitaries are compiled with a gate set consisting of CNOT and Pauli-rotation gates $\{R_X,R_Y,R_Z \}$. Each single-qubit rotation comes with a coherent over-rotation of a constant angle $\theta$. The resulting decay of the GHZ estimation fidelity, as well as the calibrated estimates from Clifford and CNOT-dihedral RB, are plotted as functions of the over-rotation angle $\theta$ in radians. In this case, the Clifford calibration performs visibly better than its CNOT-dihedral counterpart. 

These examples illustrate that there is a trade-off involved in the choice of calibration group: While the CNOT-dihedral calibration is always preferable for gate-independent noise, in more realistic scenarios involving gate-dependence, a unitary 3-design like the Clifford group may provide better performance. 

\begin{figure}[ht]
	\centering
	\includegraphics[width=0.9\textwidth]{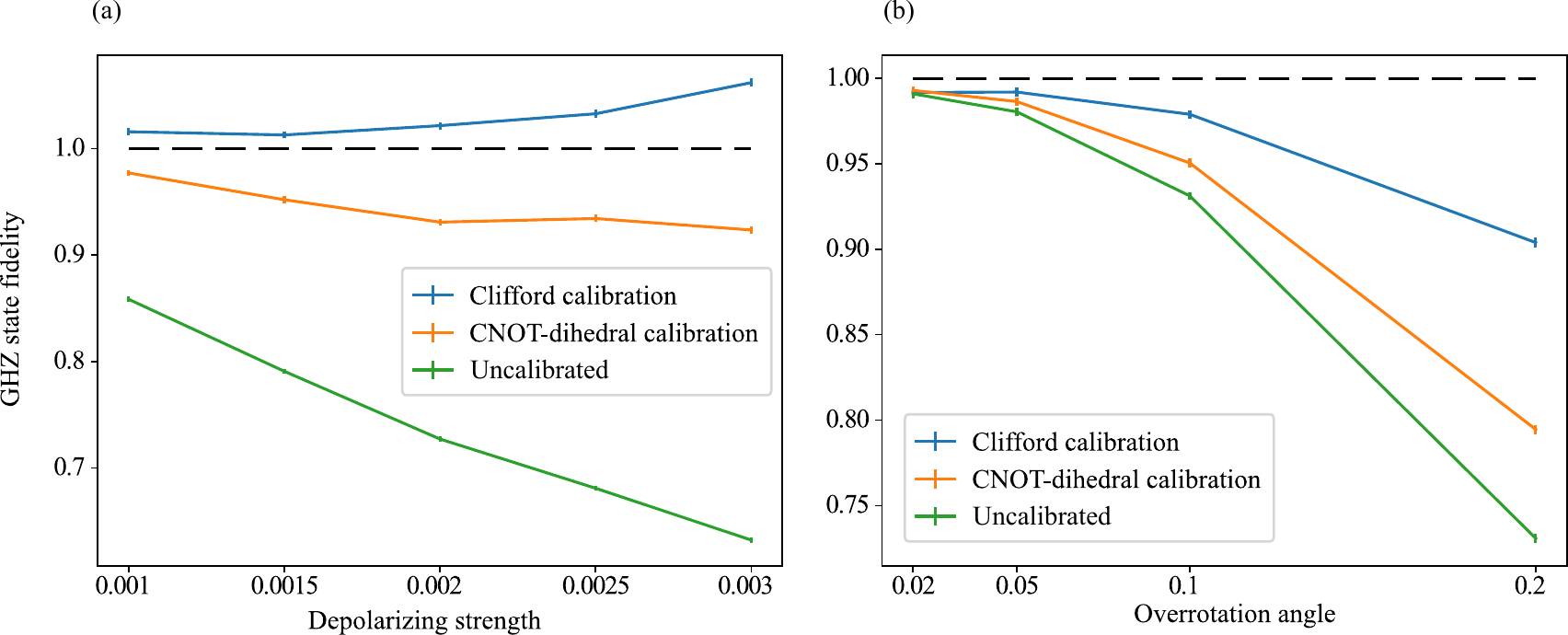}
	\caption{Performance of the calibration under gate-dependent noise for a 5-qubit system. Panel (a) uses a noise model consisting of 2-qubit depolarizing noise on CNOTs and local depolarizing noise on single-qubit gates. Panel (b) shows the performance under coherent single-qubit over-rotations when unitaries are compiled into CNOTs and single-qubit Pauli rotations.}
	\label{fig:gate_dependence}
\end{figure}

\section{Details on numerical simulations}
\label{sec:numerics_details}

The protocol has numerically been simulated using Qiskit \cite{Qiskit} and cirq \cite{cirq} for the noisy simulation of the experimental part. The post-processing has been performed with custom Python code and Qiskit's Quantum Information module. 
The task considered in Figs.\  \ref{fig:local_depol}– \ref{fig:gate_dependence} and the main text is the estimation of the state fidelity with an $n$-qubit GHZ state. The ideal GHZ state is prepared without errors, followed by noisy gates implementing the shadow estimation protocol. We then performed separate randomized benchmarking simulations with the same noise model as in the shadow estimation, and used the estimated decay parameter in the calibrated classical post-processing. If not otherwise stated, we used a total of 100000 samples for each data point shown in the figures, and chose $N=10^4$, $K=10$ for the median-of-means estimation. To estimate the statistical uncertainty, we use the bootstrapped standard deviation calculated from 200 resamplings with replacement. This standard deviation is displayed as error bars in all figures. 

Although we show that the calibration using CNOT-dihedral RB parameters is exact for gate-independent noise in the infinite-sample limit, we still observe statistical fluctuations and what appears to be systematic deviations from the ideal value, for example, in \cref{bitflip_plot}. Fluctuations around the ideal value beyond the shown error bars are explained by the statistical uncertainty in the estimation of the RB decay parameter, which we do not include in the figures. As previously noted in Ref.\ \cite{RobustShadows}, calibrated estimates exceeding the ideal value might result from the use of a ratio estimator, which introduces a bias for finite sample sizes.

\FloatBarrier

\end{document}